\documentclass[10pt,journal]{IEEEtran}

\usepackage{floatrow}
\usepackage{algorithm}
\usepackage{algpseudocode}
\usepackage{graphicx}
\usepackage{epstopdf}
\usepackage{scalefnt}
\usepackage[noadjust]{cite}
\usepackage[table]{xcolor}
\usepackage[cmex10]{amsmath}
\usepackage{amssymb}
\usepackage{multirow}

\newcommand{\qw}{{\bf w}}
\newcommand{\qW}{{\bf W}}
\newcommand{\qh}{{\bf h}}

\newtheorem{thm}{Theorem}

\newtheorem{lemma}{Lemma}

\newcommand*\rc{\color[rgb]{0, 0, 0}}
\newcommand*\rca{\color[rgb]{0, 0, 0}}
 \newcounter{subeqn} \renewcommand{\thesubeqn}{\theequation\alph{subeqn}}%
 \newcommand{\subeqn}{%
   \refstepcounter{subeqn}
   \tag{\thesubeqn}
 }

\begin{document}
\title{\rc Secrecy Analysis on Network Coding in Bidirectional Multibeam Satellite Communications}
\author{
    {\IEEEauthorblockN{Ashkan Kalantari~\IEEEmembership{Student Member,~IEEE}, Gan Zheng~\IEEEmembership{Senior Member,~IEEE}, \\ Zhen Gao~\IEEEmembership{Member,~IEEE}, Zhu Han~\IEEEmembership{Fellow,~IEEE}, 
		\\
		and Bj\"{o}rn Ottersten},~\IEEEmembership{Fellow,~IEEE}
    \\
				\thanks{Copyright (c) 2013 IEEE. Personal use of this material is permitted. However, permission to use this material for any other purposes must be obtained from the IEEE by sending a request to pubs-permissions@ieee.org}
				
		\thanks{This work was supported by the National Research Fund (FNR) of Luxembourg under AFR grant  for the project 
		``Physical Layer Security in Satellite Communications (ref. 5798109)'', SeMIGod, the National Natural Science Foundation 
		of US (Grant No. CCF-1456921, CMMI-1434789, CNS-1443917, and ECCS-1405121), and the National Natural Science Foundation of China (Grant No. 61428101).}


 \thanks{Ashkan Kalantari and Bj\"{o}rn Ottersten are with the Interdisciplinary Centre for Security, Reliability and Trust (SnT),
    The University of Luxembourg, 4 rue Alphonse Weicker, L-2721 Luxembourg-Kirchberg, Luxembourg, E-mails: ashkan.kalantari@uni.lu, bjorn.ottersten@uni.lu.}
		
		    \thanks {Gan Zheng is with the School of Computer Science and Electronic Engineering, University of Essex, UK, E-mail: ganzheng@essex.ac.uk. He is also affiliated with the Interdisciplinary Centre for Security, Reliability and Trust (SnT), University of Luxembourg, Luxembourg.}

    \thanks {Zhen Gao is with School of Electronic Information Engineering, Tianjin University, Tianjin 300072, China, Email: zgao@tju.edu.cn.}

    \thanks {Zhu Han is with the Electrical and Computer Engineering Department, University of Houston, USA, E-mail: zhan2@uh.edu.}
}}
\maketitle
\begin{abstract}
Network coding is an efficient means to improve the spectrum efficiency of satellite communications. However, its resilience to eavesdropping attacks is not well understood. This paper studies the confidentiality issue in a bidirectional satellite network consisting of two mobile users who want to exchange message via a multibeam satellite  using the XOR network coding protocol. We aim to maximize the sum secrecy rate by designing the optimal beamforming vector along with optimizing the return and forward link time allocation. The problem is non-convex, and we find its optimal solution using semidefinite programming together with a 1-D search. For comparison, we also solve the sum secrecy rate maximization problem for a conventional reference scheme without using network coding. Simulation results using realistic system parameters demonstrate that the bidirectional scheme using network coding provides considerably higher secrecy rate compared to that of the conventional scheme.
\end{abstract}
\section{Introduction}
Satellite communications (SATCOM) is getting  more and more
integrated into communication networks to compliment the current
terrestrial communication systems. Satellite services have to
support increasing demands for data transfer.  To realize
bidirectional satellite communications, traditionally orthogonal
resources either in frequency or time domain should be used to avoid
interference between users. To save the precious wireless resources,
network coding has been used in this work as an efficient protocol to
exchange information between two mobile satellite users. The basic
principle is that the received information from users are combined
on the satellite or gateway (GW), and then the mixed signal is broadcast to users at the
same time and using the same frequency. Because each user can
subtract its own message, it can easily decode the message from the
other user.

However, due to the broadcast nature and immense area coverage,
satellite communications systems, e.g., in military and commercial
applications, are vulnerable to security attacks such as
eavesdropping. Currently, security in SATCOM is achieved at upper
layers by means of encryption such as the Advanced Encryption
Standard~\cite{Chowdhury:2005,Cruickshank:2005}. Nevertheless,
traditional security is based on the assumption of limited
computational capability of the malicious nodes, and thus there exists
the risk that a malicious node can successfully break an encryption,
and get access to sensitive satellite data
\cite{sklavos:Cryptography:2007}. In contrast to the upper layer
encryption techniques, recently there has been significant interest
in securing wireless communications at the physical layer using an
information-theoretic approach named ``\emph{secrecy rate}''~\cite{Wyner:1975}. The main advantage of this approach
is that the malicious nodes cannot even get access to protected
information regardless of their computational capabilities.

While network coding can greatly improve the system throughput,
whether it is more secure than the conventional scheme, which does not 
use network coding, is largely unknown in SATCOM. In this work, we
will leverage the physical layer security approach to address the
confidentiality issue in bidirectional SATCOM using the principle of
network coding. Below, we provide an overview on the applications of
network coding to SATCOM and the related work in the physical layer
security literature.

\subsection{Literature Review}
\subsubsection{Network coding related works}
Network coding technique, first introduced in~\cite{Rudolf:2000},
can considerably reduce delay,  processing complexity and power
consumption, and can significantly increase the data rate and
robustness~\cite{Bassoli:2013}. In the popular XOR network coding
scheme, the received signals at an intermediate node are first
decoded into bit streams, and then XOR is applied on the bit streams
to combine them. The processed bits are re-encoded and then
 broadcast. Utilization of network coding has been studied in both
terrestrial and satellite networks. The authors
in~\cite{Hammerstrom:2007} apply superposition coding and XOR
network coding to a bidirectional terrestrial relay network. A
multi-group multi-way terrestrial relay network is considered
in~\cite{Amah:2011} where superposition coding and XOR network
coding are investigated and compared to each other. Network coding
can also considerably improve the spectral efficiency in
bidirectional SATCOM in which two mobile users exchange information
via the satellite. The work in~\cite{Rossetto:2010} compares the
amplify-and-forward (AF) method with the XOR network coding
scheme in a satellite scenario. A joint delay and packet drop rate
control protocol without the knowledge of lost packets for mobile
satellite using network coding is studied in~\cite{Sorour:2010}.
In~\cite{Rossetto:2011}, buffers are designed for satellites when
the network coding scheme is employed. Random linear network coding
is used in~\cite{Vieira:2012:Load} to minimize the packet delivery
time. Satellite beam switching for mobile users is tackled
in~\cite{Vieira:2012:Handover} where  the network coding scheme
increases the robustness in delivery of the packets when mobile
terminals move from beam to beam. The XOR network coding protocol is
demonstrated in a satellite test bed in~\cite{Bischl:2011}.

\subsubsection{Physical layer security related works}
Wyner in~\cite{Wyner:1975} first showed that secure transmission is
possible for the legitimate user given the eavesdropper receives
noisier data compared to the legitimate receiver. Inspired by
Wyner's work,~\cite{Leung-Yan-Cheong:1978} extended the idea of physical layer secrecy rate from
the discrete memoryless wiretap channel to Gaussian wiretap channel.
The Wyner's wiretap channel was generalized in~\cite{Csiszar:1978}
to the broadcast channel. After the seminal works done
in~\cite{Wyner:1975,Leung-Yan-Cheong:1978,Csiszar:1978}, there have
been substantial amount of works in physical layer secrecy. Here, we
only review those most relevant to network coding and bidirectional
communications. The authors in~\cite{Jingchao:2011} consider a relay
utilizing the XOR network coding protocol where joint relay and
jammer selection is done to enhance the secrecy rate. {\rc A
bidirectional AF relay network with multiple-antenna nodes is
considered in~\cite{Mukherjee:2010} where the relay beamforming
vector is designed by the waterfilling method to improve the secrecy
rate}. The authors in~\cite{Zhiguo:2011} consider random relay
selection in a bidirectional network in which the relay performs 
both data transmission and jamming the eavesdropper at the same time to increase the secrecy.
 The work in~\cite{Zhiguo:2012} performs selection over AF relays and
 jammers in a bidirectional network for the single-antenna case, and
 precoding in the multiple-antenna case to enhance the secrecy. To
maximize the secrecy in a bidirectional network, the authors
in~\cite{Jingchao:2012} consider the location and distribution of
nodes while joint relay and jammer selection is performed.
Distributed beamforming along with artificial noise and beamforming
is studied in~\cite{Hui-Ming:2012} for a bidirectional AF relay
network. The work in~\cite{Hui-ming:2013} designs the distributed
beamforming weights for a bidirectional network where one
intermediate node acts as a jammer. In contrast to the terrestrial
literature, there are very few works in physical layer security for
SATCOM. The problem of minimizing the transmit power on a multibeam
satellite  while satisfying a minimum per user secrecy rate is
studied in~\cite{Lei:2011}. Iterative algorithms are used to jointly
optimize the transmission power and the beamforming vector by
perfectly nulling the received signal at the eavesdropper. Both
optimal and suboptimal solutions are developed in~\cite{Gan:2012}
where   the use of artificial noise is also studied.

Despite the physical layer security and network coding works in the
terrestrial and SATCOM scenarios, some unaddressed issue are left. 
In~\cite{Hammerstrom:2007}, only downlink bottlenecks
are considered when designing the beamforming weights for the XOR
network coding case. The uplink bottlenecks also need to be considered
when optimizing the uplink-downlink time allocation.
In~\cite{Amah:2011}, the authors consider the decoding-re-encoding
and designing the beamforming vector separately. {\rc The works
in~\cite{Zhiguo:2011,Jingchao:2012}
consider single-antenna relay where the AF protocol is used in a
bidirectional network. The authors in~\cite{Mukherjee:2010,Hui-Ming:2012,Hui-ming:2013} use the analog network coding protocol in a two-way relay network to facilitate secure information exchange between two users.} Furthermore, the mentioned terrestrial works in physical layer
security for bidirectional communications assume one eavesdropper in
the environment. The works in~\cite{Lei:2011,Gan:2012} design the
beamforming weights for unidirectional service for fixed users in
the forward link (FL).

\subsection{Our Contribution}
In this work, we study the network coding based bidirectional SATCOM
in which two mobile users  exchange data via a transparent multibeam
satellite in the presence of two eavesdroppers. There is an
eavesdropper present for each user who overhears the bidirectional
communications. The users employ omnidirectional antennas and the
communication is prone to eavesdropping in both the return link (RL)
and FL. In the RL, two users send signals using two orthogonal
frequency channels; the signals collected by the satellite are
passed to the GW, where they are decoded, XOR-ed and 
then the produced stream is re-encoded. This combined stream is {\rca multiplied by the beamforming vector which 
contains the designed weight of each feed. Consequently, each element of the resultant vector is transmitted 
to the satellite using the feeder link. Each element which includes both the feed weight and the 
data signal is applied to the corresponding feed to adjust the beams} for broadcasting to both users simultaneously
in the FL. This scheme is more power-efficient than the conventional
method where network coding in not utilized and the power is splitted into two
data streams. This benefit is extremely vital for SATCOM because of
the limited on-board power.

Our main contributions in this work are summarized below to
differentiate it from the prior work:
\begin{enumerate}
  \item We incorporate XOR network coding into SATCOM in order to
  enable both efficient and secure
  bidirectional data exchange. 

  \item The end-to-end sum secrecy rate is first derived, and then maximized by designing the optimal beamforming vector and the RL and FL time allocation. The optimization
  problem regarding the beamforming vector is solved using semi-definite programming (SDP) along with 1-D search.

  \item We provide comprehensive  simulation results to demonstrate the advantage of the bidirectional scheme over the conventional scheme using realistic SATCOM parameters.
\end{enumerate}

The remainder of the paper is organized as follows. In
Section~\ref{Sec:System model}, we introduce the SATCOM network
topology as well as deriving the signal model and defining the secrecy rates.
The  problems for maximizing the sum secrecy rate are defined and
solved in Section~\ref{Sec:Problem Formulation}. In
Section~\ref{Sec:Simulation Results}, numerical results  are
presented. The conclusion is drawn in Section~\ref{sec:con}.

\hspace{-0.3cm}\emph{Notation}: Upper-case and lower-case bold-faced
letters are used to denote matrices and column vectors,
respectively. Superscripts $(\cdot)^T$ and $(\cdot)^H$ represent
transpose and Hermitian operators, respectively. ${{\bf{I}}_{N
\times N}}$ denotes an $N$ by $N$ identity matrix. $\mathcal{CN}(
\bf{m},\bf {K})$ denotes the complex Gaussian distribution with mean
vector $\bf{m}$ and covariance matrix $\bf {K}$. $\lambda _{\max
}(\bf A, \bf B)$ is the maximum eigenvalue of the matrix pencil
$(\bf A, \bf B)$. $\bf{A}\succeq \bf{0}$ means that the Hermitian
matrix $\bf{A}$ is positive semidefinite.  $\|\cdot\|$ is the
Frobenius norm and $|\cdot|$ represents the absolute value of a
scalar.
\section{System Model}\label{Sec:System model}
Consider a satellite communication system comprised of two users
denoted by $U_1$ and $U_2$ who exchange information with each other,
one multibeam transparent satellite denoted by $S$, one GW, two
eavesdroppers denoted by $E_1$ and $E_2$ as depicted in
Fig.~\ref{Fig:System model}. Users are located in different beams of
the satellite, and they transmit the RL signals using different
frequency channels simultaneously. We assume that each user and each
eavesdropper is equipped with a single omni-directional antenna.
Because of the long distance between the users, there is no direct
link between them; furthermore, eavesdroppers cannot cooperate and
$E_i$ can only overhear $U_i$ for $i=1,2$. Contemporary {\rc orbiting} satellites 
{\rc such as \emph{ICO}, \emph{SkyTerra}, and \emph{Thuraya}} have limited power, here defined as $P_S$, and some of
them do not have the on-board processing ability to decode the
received messages {or \rca perform on-board beamforming,} so they have to forward the received signal to
the GW to get it processed{\rc~\cite{Sunderland:2002,Devillers:2011:joint,Khan:2012}}. {\rc Using the GW to process the signal and 
designing the feed weights is referred to as the ground-based beamforming technique. The 
ground-based beamforming technique is perceived as the most convenient and economical 
approach~\cite{Khan:2012}}. {\rca In this paper, we consider a commercial satellite without digital processing ability and 
follow the ground-based beamforming paradigm.}

{\rc In our satellite network model, we assume that the eavesdropper is a regular user which is part of the network. However, it is considered as an unintended user, potential eavesdropper, which the information needs to be kept secret from it. Due to the fact that the eavesdropper is part of the network, it is possible to estimate the channels to it. Hence, similar to the works~\cite{Tekin:2008,Tie:2009,Dong:2010,Yuksel:2011,Hang:2013}, we assume that the eavesdropper's channel state information (CSI) is 
known. Based on the mentioned assumption, we assume that the users and eavesdropper know all the CSIs}. Further, all communication
channels are known and fixed during the period of communication. {\rca It is worth mentioning that in the secrecy rate analysis of 
XOR network coding, only the CSI of the eavesdroppers in the RL is required. Although we assume the availability of the eavesdropper's CSI, there are methods such as null-space 
artificial noise transmission~\cite{Goel:2008}, random beamforming~\cite{Qiaolong:2013,Chorti:2012:main,Xiaohua:2007}, or effective channel coding design to 
strengthen the cryptography~\cite{Harrison:2009} in order to sustain secrecy without having the knowledge of the eavesdropper's CSI. Another alternative can be using the statistical knowledge of the eavesdropper's CSI in order to improve the secrecy~\cite{Jiangyuan:2011,Jiangyuan:G:2011,Rezki:2014,Girnyk:2014}. Also, the interference alignment technique can be used along with statistical knowledge of the eavesdropper's CSI to enhance the secrecy~\cite{Koyluoglu:2011}. In the situations when the geographical area of the eavesdropper is known, the worst-case scenario can be considered. In this scenario, the best CSI from the user to the eavesdropper's area is considered for the design. One possible example for the worst-case scenario can be when the occupied zone by the enemy is known. This example can be one of the applications of this paper.} 

To acquire the RL channel state information (CSI) at the GW, the users
send the pilot signals along with the data toward the satellite. For
the FL CSI, the GW sends pilots to the users. Afterwards, the estimated CSI
by the users is sent back to the GW.
Therefore, getting the FL CSI takes more time compared to the RL
CSI~\cite{Byoung-Gi:2006}. The GWs are equipped with advanced transceivers and antennas and because of this reason, the communication link between the GW and the 
satellite (feeder link) is modeled as an ideal link. Hence, {\rc similar to the works~\cite{Cottatellucci:2006,Devillers:2011:joint,Christopoulos:2012,arnau:2012,Joroughi:2013} which 
are carried out in the satellite communications literature,} we assume that the channel
between the satellite and the GW, which is referred to as the feeder link, is ideal with abundant bandwidth.

The complete communication phases of the network coding
based scheme are summarized in Table~\ref{table:Phases}. The
conventional scheme without using network coding is also included
for comparison and details are given in Section~\ref{con:satcom}. The first two
phases for the RL are the same for both schemes while the main
difference lies in the FL transmission. In the conventional scheme,
signals are sent in different time slots for each user in the FL, so
this scheme has less available transmission time for each user. In
the bidirectional scheme, signal streams are combined, and then
  sent in the FL using the XOR network coding protocol, therefore, 
	the spectral efficiency is significantly improved compared
to the conventional scheme.

\begin{figure}[t]
  \centering
  \includegraphics[width=8.5cm]{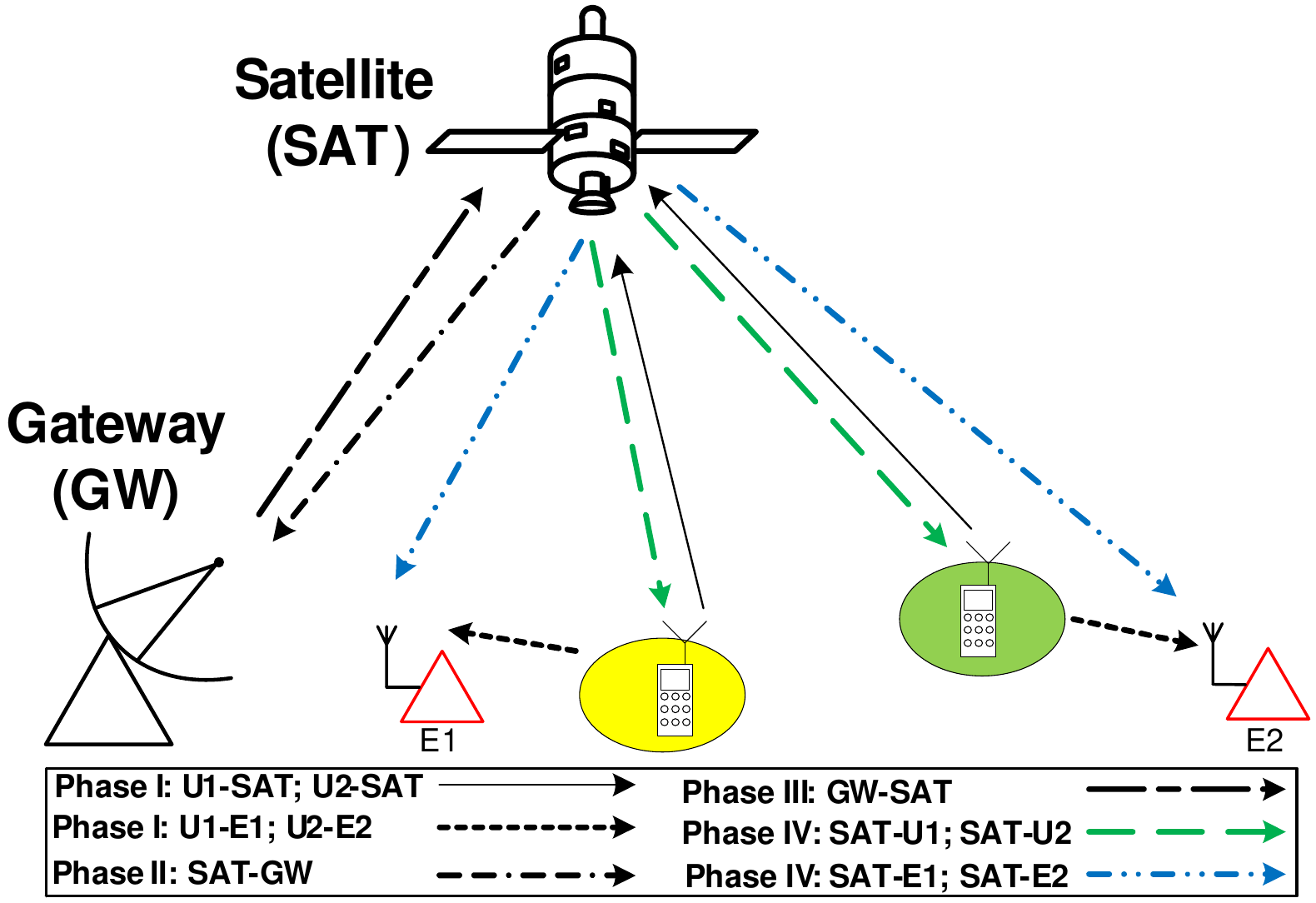}
  \caption{Bidirectional satellite communication network.}
  \label{Fig:System model}
\end{figure}
\begin{table}[]
 \centering
\caption{Communication stages for the XOR network coding and the conventional schemes.}
\label{table:Phases}
\begin{tabular}{|p{4cm}|p{4cm}|}
\hline
 \small Conventional reference scheme  &  \small XOR network coding scheme     \\ \hline
\multicolumn{2}{|p{8cm}|}{\small Phase I: $U_1$ and $U_2$ simultaneously send their signals, $s_1$ and $s_2$, to the satellite
while they are overheard by $E_1$ and $E_2$, respectively.} \\ \hline

\multicolumn{2}{|p{8cm}|}{\small Phase II: The satellite passes the received signal to the gateway for
processing. At the gateway, the users’ signals are separately decoded.} \\ \hline

\small Phase III: The intended signal for $U_1$, decoded $s_2$, is re-encoded at the gateway {\rca and the corresponding 
feed weights are designed. Then, the feed weights multiplied by the data signal} are sent to the satellite.                                       & \multirow{2}{*}{\begin{tabular}[c]{p{3.9cm}} 
\small Phase III: The gateway applies XOR operation on the decoded streams from $s_1$ and $s_2$ to create a merged stream of bits {\rca and the 
feed weights are designed. Then, the feed weights multiplied by the data signal are sent to the satellite. } \end{tabular}}  \\ \cline{1-1}

\small Phase IV: The satellite passes the re-encoded signal through the corresponding beam to $U_1$ while $E_1$ is listening to it.               &      \\ \hline

\small Phase V: The intended signal for $U_2$, decoded $s_1$, is re-encoded at the gateway {\rca and the corresponding 
feed weights are designed. Then, the feed weights multiplied by the data signal} are sent back to the satellite.                                  & \multirow{2}{*}{\begin{tabular}[c]{p{3.9cm}} 
\small Phase IV: The satellite broadcasts the merged stream toward the users through the corresponding beams which is wiretapped by both $E_1$ and $E_2$. \end{tabular}}   \\ \cline{1-1}

\small Phase VI: The satellite passes the re-encoded signal through the corresponding beam to $U_2$ while $E_2$ is listening to it.               &                 \\ \hline
\end{tabular}
\end{table} 
\subsection{Network coding based bidirectional SATCOM}\label{Sec:Signal model Bi}
\subsubsection{Signal model}\label{Sec:Signal model Bi sub}
In this case, the whole communication takes place in four phases. In
Phase I, both users transmit signals using different frequencies
simultaneously. The signals received at the satellite and the eavesdroppers are
\begin{align}
&{{\bf{y}}_{S_1}} =  \sqrt {P_{U_1}} {\qh_{{U_1},S}}\,{s_1} + {{\bf{n}}_{S_1}},
\label{eqn:SAT rec sig PhI U1}
\\
&{{\bf{y}}_{S_2}} =  \sqrt {P_{U_2}} {\qh_{{U_2},S}}\,{s_2}+  {{\bf{n}}_{S_2}},
\label{eqn:SAT rec sig PhI U2}
\\
&y_{{E_1}}^{RL} =  \sqrt {P_{U_1}} {h_{{U_1},{E_1}}}{s_1} + n_{{E_1}},
\label{eqn:E1 rec sig PhI}
\\
&y_{{E_2}}^{RL} =  \sqrt {P_{U_2}} {h_{{U_2},{E_2}}}{s_2} + n_{{E_2}},
\label{eqn:E2 rec sig PhI}
\end{align}
where $P_{U_i}$ is the transmitted power by the users for $i=1, 2$,
$h$ and $\bf{h}$ represent the user-eavesdropper and user-satellite
channels, respectively, and the corresponding source and destination
are denoted by the subscript. The channel for the satellite is a 
$N_S \times 1$ vector where $N_S$ is the number of the satellite 
feeds. Additive white Gaussian noises (AWGN) are denoted by
$n$ and $\textbf{n}$ with $n\sim\mathcal{CN}(0,\sigma^2)$ and
$\bf{n}\sim\mathcal{CN}(0,\sigma^2{{\bf{I}}}_{N_S\times N_S})$,
respectively. We consider the noise power for users, satellite and
eavesdroppers as $KTB$, where $K$ is the Boltzman's constant which
is $-226.8$ dBW/K/Hz, $T$ is the on-board temperature and $B$ is the
carrier bandwidth.   We assume that $s_1$ and $s_2$ are independent
and identically distributed (i.i.d.) Gaussian random source signals
with zero mean and unit variance. For convenience, we use the noise
variance, $\sigma^2$, instead of $KTB$ and omit the bandwidth, $B$,
in the rate expressions throughout the paper. Note that we consider
different temperatures for ground nodes and the satellite. The
satellite forwards the received signal to the GW {\rca using the feeder link} in Phase II and
thanks to the ideal link between the satellite and the GW, the same
signals as~\eqref{eqn:SAT rec sig PhI U1} and~\eqref{eqn:SAT rec sig
PhI U2} are present at the GW to be processed.

At the GW, the received signal is filtered and users' data are
separated and decoded into two bit streams denoted by $x_1$ and
$x_2$, respectively. The GW applies the bit-wise XOR algebraic
operation to the decoded bit streams of the users to get the
combined stream
\begin{align}
{x_{GW}} = {x_1} \oplus {x_2}.
\label{eqn:XOR GW Process}
\end{align}
Note that before applying the XOR network coding, the GW uses
zero-padding to {\rc add zeros to the shorter bit stream in order to} 
make equal length {\rc bit streams} out of the {\rc two} different bit
streams sent by the users{\rc~\cite{Katti:2008,Xuehua:2013}.} In Phase III, $x_{GW}$ is encoded into
$s_{GW}$ with unit power, and then multiplied by the beamforming vector,
$\bf{w}$. {\rca Using the ideal feeder link, each element, $w_is_{GW}(t)$, of the produced vector, ${\bf{w}}s_{GW}$, at 
the GW which both includes the feed weight, $w_i$, and the data signal, $s_{GW}$, is transmitted from the GW to the satellite.} Since the codebook used at the
GW to encode $x_{GW}$ can be different in the XOR network coding
scheme, the RL and FL transmission times are generally different for
the XOR network coding. This enables optimum RL and FL time
allocation for the XOR network coding. The received signal by
satellite is denoted as ${{\bf{s}}_S} = {\rca {\bf{H}}_{GW,S}} \qw{s_{GW}}$. {\rca The model ${\bf{s}}_S = {\bf{H}}_{GW,S} {\bf{w}} s_{GW}$ 
encapsulates the process of transmitting each element of the vector 
${\bf{w}}s_{GW}$ from the GW to the satellite. Since the feeder link is considered to be ideal, ${\bf{H}}_{GW,S}$ is a $N_S \times N_S$ identity 
matrix.} Finally, in Phase IV, {\rca each feed weight designed at the gateway, which includes the data signal, 
is applied to the corresponding feed at the satellite. Hence, the beams are adjusted} 
and the signal ${{\bf{s}}_S}$ is broadcast {\rca through the antennas}. 
The received signals at two users are, respectively,
\begin{align}
y_{{U_1}}^{FL_{XOR}} &= {\qh^T_{S,{U_1}}}{{\bf{s}}_S} + {n_{{U_1}}},
\label{eqn:XOR U1 downlink sig}
\\
y_{{U_2}}^{FL_{XOR}} &= {\qh^T_{S,{U_2}}}{{\bf{s}}_S} +
{n_{{U_2}}}. \label{eqn:XOR U2 downlink sig}
\end{align}
 Similarly, the  received signals at the eavesdroppers  in Phase IV
 are, respectively,
\begin{align}
y_{{E_1}}^{FL_{XOR}} & =  {\qh^T_{S,{E_1}}}{{\bf{s}}_S} + n_{{E_1}},
\label{eqn:XOR E1 downlink rec sig}
\\
y_{{E_2}}^{FL_{XOR}} & =  {\qh^T_{S,{E_2}}}{{\bf{s}}_S} + n_{{E_2}}.
\label{eqn:XOR E2 downlink rec sig}
\end{align}
In the following, we shall define the sum secrecy rate. We first
introduce the users' rates and eavesdroppers' channel capacities.
\subsubsection{Users' RL rates}\label{Subsec:Uplink Rates}
Consider $t_1$ and $t_2$ for the RL (Phase I) and FL (Phase IV) transmission
time, respectively. In Phase I, we can characterize the RL rates $(R^{RL}_{U_1},R^{RL}_{U_2})$
by the following equations~\cite[Chapter~5]{Tse:WirelessComm}:
\begin{align}
& R_{{U_1}}^{RL} \le I_{{U_1}}^{RL} = {t_1}\log \left( {1 + \frac{{{P_{{U_1}}}{{\left\| {{{\bf{h}}_{{U_{1,S}}}}} \right\|}^2}}}{{\sigma _S^2}}} \right)
\label{eqn:U1 Uplink Rate}
\\
& R_{{U_2}}^{RL} \le I_{{U_2}}^{RL} = {t_1}\log \left( {1 + \frac{{{P_{{U_2}}}{{\left\| {{{\bf{h}}_{{U_{2,S}}}}} \right\|}^2}}}{{\sigma _S^2}}} \right),
\label{eqn:U2 Uplink Rate}
\end{align}
where $I$ denotes channel capacity or the maximum supported rate and $R$
is the maximum achievable rate.
\subsubsection{Users' FL rates}\label{Para:XOR User Rate}
After receiving the FL signal, users decode ${{\bf{s}}_S}$. As each
user knows its own transmitted bits, it can use the XOR operation to
retrieve the intended bits. Subsequently, using~\eqref{eqn:XOR U1
downlink sig} and~\eqref{eqn:XOR U2 downlink sig}, the FL rates can
be expressed as
\begin{align}
& R^{FL_{XOR}} = \min \left\{ {I_{{U_1}}^{FL_{XOR}},I_{{U_2}}^{FL_{XOR}}} \right\},
\label{eqn:XOR downlink rate}
\\
& I_{U_1}^{FL_{XOR}} =  {t_2}\log \left( {1 +\frac{|{\bf{h}}_{S,{U_1}}^T \bf{w} |^2}{\sigma_{{U_1}}^2}} \right),
\label{eqn:XOR U1 downlink rate region}
\\
& I_{U_2}^{FL_{XOR}}
=  {t_2}\log \left( {1 +\frac{|{\bf{h}}_{S,{U_2}}^T \bf{w} |^2}{\sigma_{{U_2}}^2}} \right).
\label{eqn:XOR U2 downlink rate region}
\end{align}
 Since the data for both users have gone through a bit-wise XOR
operation at the GW and a combined signal is broadcast, the GW has
to adjust the combined signal's data rate to match both users'
channel capacities. This rate should be equal to the minimum FL
channel rate between the satellite and the users in Phase IV before
sending ${{\bf{s}}_S}$ to the satellite.
\subsubsection{Eavesdroppers' channel capacities}\label{Para:XOR Eve rate}
Using~\eqref{eqn:E1 rec sig PhI} and~\eqref{eqn:XOR E1 downlink rec
sig}, the channel capacity from $U_1$ to $E_1$, $I_{E_1}^{RL}$, and
from satellite to $E_1$, $I_{E_1}^{F{L_{XOR}}}$, can be expressed,
respectively, as
\begin{align}
&I_{E_1}^{RL} = {t_1}\log \left( {1 + \frac{  P_{U_1} {\left| h_{ U_1 , E_1 } \right|}^2} {\sigma _{{E_1}}^2}} \right),
\label{eqn:XOR E1 RL rate}
\\
&I_{E_1}^{FL_{XOR}}
=  {t_2}\log \left( {1 +\frac{|\qh_{S,{E_1}}^T {\bf{w}} |^2}{\sigma_{{E_1}}^2}} \right).
\label{eqn:XOR S-E1 FL rate}
\end{align}
{\rca The channel capacities for $E_2$ can be derived in a similar way.}
\subsubsection{Secrecy rate definition} \label{sec:sec rate def bi}
First, we derive the the secrecy rate for the RLs and FLs, {\rca and then the end-to-end secrecy rate.} 
{\rca In~\cite{Oggier:2008}, the result of~\cite{Wyner:1975} is extended to fading channels with multiple-antenna transmitter, receiver, and eavesdropper}. {\rca Using the special case of the result in~\cite{Oggier:2008} for single-antenna transmitter, multiple-antenna receiver, and single-antenna eavesdropper along with employing}~\eqref{eqn:U1 Uplink Rate} and~\eqref{eqn:XOR E1 RL rate}, the
secrecy rate for the RL of $U_1$ is calculated as
\begin{align}
{\rca SR_{{U_1}}^{RL}} = I_{{U_1}}^{RL} - I_{{E_1}}^{RL},
\label{eqn:XOR secrecy rate RL U1}
\end{align}
{\rca
where the notation ``SR'' means ``secrecy rate''.}

{\rca To calculate the secrecy rate in the FL, first, we derive the information that $E_1$ can recover 
during the RL transmission in Lemma~\ref{lem:E1 Ph IV}.}

{\rca  \begin{lemma}
Independent of getting a positive or zero secrecy rate defined for the RL of $U_1$ in~\eqref{eqn:XOR secrecy rate RL U1}, $E_1$ 
cannot recover any bits from $U_2$ transmitted message using the FL transmission.

\begin{proof}
{\rca To recover bits from $U_2$, $E_1$ has to apply XOR operation between the bits recovered from $U_1$ in the RL transmission and 
the bits derived from the satellite broadcast in the FL transmission. Hence, the information detected by $E_1$ in the FL depends on the bits recovered from $U_1$ in the RL transmission. The recovered bits from $U_1$ in the RL depend on the 
sign of the secrecy rate defined in~\eqref{eqn:XOR secrecy rate RL U1}. The sign of the RL secrecy rate in~\eqref{eqn:XOR secrecy rate RL U1} has the following possibilities:}

\begin{enumerate}

	\item If $I_{{U_1}}^{RL} - I_{{E_1}}^{RL}>0$, then $U_1$ can establish a perfectly secured connection so that the 
	eavesdropper cannot get any bits from $U_1$ in the RL~\cite{Oggier:2008}. Hence, $E_1$ does not have the bits 
	transmitted by $U_1$ in the RL and it cannot recover any bits from $U_2$ using the FL transmission.
		
	\item If $I_{{U_1}}^{RL} - I_{{E_1}}^{RL} \le 0$, then the secrecy rate is zero. Therefore, $U_1$ cannot 
	establish a secure connection in the RL. In this case, $U_1$ remains silent during the corresponding time slot. In this time slot, GW generates random bits instead of the bits from $U_1$ and applies XOR between them and the bits from $U_2$. As a result, $E_1$ cannot recover any bits from $U_2$ using the FL transmission.

\end{enumerate}
Note that since the RL time, $t_1$, is always positive and all the channels are known, the sign of the 
expression $I_{{U_1}}^{RL} - I_{{E_1}}^{RL}$ is known prior to the beamformer design. 
\end{proof}
\label{lem:E1 Ph IV}
\end{lemma}
A similar argument as in Lemma~~\ref{lem:E1 Ph IV} can be applied to $E_2$. 

Consequently, using Lemma~\ref{lem:E1 Ph IV}, the secrecy rate for the FL is given in Lemma~\ref{lem:FL SR}.}

{\rca  \begin{lemma}
Assume that there exists at least one RL with a positive secrecy rate. Then, the secrecy rate in 
the FL is given as below: 
\begin{align}
S{R^{F{L_{XOR}}}}   = \left\{ {\begin{array}{*{20}{c}}
 \min \left\{ {I_{{U_1}}^{F{L_{XOR}}},I_{{U_2}}^{F{L_{XOR}}}} \right\} & SR_{{U_1}}^{RL}>0,\\
& SR_{{U_2}}^{RL}>0,
\\
\\
I_{{U_1}}^{F{L_{XOR}}}                                                & SR_{{U_1}}^{RL}=0,\\
& SR_{{U_2}}^{RL}>0,
\\
\\
I_{{U_2}}^{F{L_{XOR}}}                                                & SR_{{U_1}}^{RL}>0,\\
& SR_{{U_2}}^{RL}=0.
\end{array}} \right.
\label{eqn:FL SR} 
\end{align}
\begin{proof}
Excluding the case that both RLs have zero secrecy rate, i.e., the total secrecy rate is zero, the secrecy 
rate for the FL transmission for different signs of the secrecy rate in the RL is given as follows:

\begin{enumerate}
	
	\item If $SR_{{U_1}}^{RL}>0$ and $SR_{{U_2}}^{RL}>0$, then according to Lemma~\ref{lem:E1 Ph IV}, 
	$E_1$ and $E_2$ cannot wiretap any bits from $U_2$ and $U_1$, respectively, using the FL transmission. Therefore, 
	using~\eqref{eqn:XOR downlink rate}, the secrecy rate in the FL is 
	$\min \left\{ {I_{{U_1}}^{F{L_{XOR}}},I_{{U_2}}^{F{L_{XOR}}}} \right\}$. \label{c:1}
	
	\item If $SR_{{U_1}}^{RL}>0$ and $SR_{{U_2}}^{RL}= 0$, then according to Lemma~\ref{lem:E1 Ph IV}, 
	$E_1$ cannot wiretap any bits from $U_2$ using the FL transmission. Further, since the RL of $U_2$ is not secure, $U_2$ does 
	not transmit and $E_2$ does not get any bits from $U_2$. Hence, $E_2$ cannot recover bits from $U_1$ using the FL 
	transmission. Since $U_1$ is not expected to receive any message because of $SR_{{U_2}}^{RL}= 0$, the FL secrecy rate is	$I_{{U_2}}^{F{L_{XOR}}}$.\label{c:2}

  \item If $SR_{{U_1}}^{RL}=0$ and $SR_{{U_2}}^{RL}>0$, similar to the procedure as in Case~\ref{c:2}, the secrecy rate in the FL is 
	$I_{{U_1}}^{F{L_{XOR}}}$.\label{c:3}
\end{enumerate}
According to the results in Cases~\ref{c:1},~\ref{c:2}, and~\ref{c:3}, the secrecy rate of the FL is derived as in~\eqref{eqn:FL SR}. 
\end{proof}
\label{lem:FL SR}
\end{lemma}

According to Lemma~\ref{lem:FL SR}, when the XOR protocol is used, the FLs are totally secured. Note 
that for the Cases~\ref{c:2} and~\ref{c:3}, the GW creates random bits instead of the 
message from the user with insecure link, i.e., zero secrecy rate in the RL. Then, the GW applies XOR between the 
received message from the user which has a positive secrecy rate in the RL and the randomly generated 
bits. This way, the eavesdropper still receives a combined message when the secrecy rate is 
zero in one of the RLs.}

To derive the end-to-end secrecy rate for $U_1$, we invoke
Theorem~1 in~\cite{Awan:2012}, which states that, when decoding
and re-encoding is performed by an intermediate node, the secrecy
rate of each hop needs to be taken into account as a bottleneck to
derive the end-to-end secrecy rate. Since decoding and re-encoding
is performed at the GW, the result of Theorem~1 in~\cite{Awan:2012} can be applied. 
Consequently, using the mentioned theorem and the secrecy rate
derived in~\eqref{eqn:XOR secrecy rate RL U1} {\rca and the result of Lemma~\ref{lem:FL SR} in~\eqref{eqn:FL SR}}, 
the end-to-end secrecy rate for $U_1$ is
calculated by
{\rca \begin{align}
 SR_{{U_1}}^{XOR} &= {\rm{ }}\min \left\{ {SR_{{U_1}}^{RL},SR_{{U_1}}^{F{L_{XOR}}}} \right\}.                  
\label{eqn:XOR SR U1}
\end{align}}
The end-to-end secrecy rate for $U_2$ can be derived in a similar
way. The sum end-to-end secrecy rate is expressed as
\begin{align}
{ SR^{XOR} = SR_{{U_1}}^{XOR} + SR_{{U_2}}^{XOR}.}
\label{eqn:XOR sum secrecy rate}
\end{align}
\subsection{Conventional SATCOM} \label{con:satcom}
A   conventional scheme without using  network coding is described
here as a performance benchmark.
\subsubsection{Signal model}
As shown in Table~\ref{table:Phases}, the Phases I and II are the
same for the conventional and the XOR network coding
schemes, which result in the same signal model for both schemes. 
In Phases III and V, the GW sends back {\rca each element of the} processed $\textbf{s}_2$ and
$\textbf{s}_1$ to the satellite{\rca , respectively, using the ideal feeder link} where $\textbf{s}_1$ and
$\textbf{s}_2$ are $N_S\times 1$ vectors {\rca containing both the feed weights and the users' data signals}. $\textbf{s}_1$ and
$\textbf{s}_2$ are defined as ${{\bf{s}}_1} = {\qw_1}{\hat{s}_1}$
and ${{\bf{s}}_2} = {\qw_2}{\hat{s}_2}$, where ${\hat{s}_1}$ and
${\hat{s}_2}$ are {\rc the decoded and re-encoded versions of the data signals 
received from $U_1$ and $U_2$ at the GW
with unit power, and $\textbf{w}_1$ and $\textbf{w}_2$ are beamforming
vectors to be designed at the GW. Note that since different Gaussian codebooks are used at the 
GW to re-encode the signals for $U_1$ and $U_2$, the generated signals at the GW are different 
from those received from the users. Therefore, generated signals at the GW are shown 
by ${\hat{s}_1}$ and ${\hat{s}_2}$.} 

{\rc The satellite applies {\rca each component of the vector ${{\bf{s}}_2}$, containing the feed weight multiplied by the data signal, to the corresponding 
feed. Then, the beam is adjusted and}  $\hat{s}_2$ is sent toward $U_1$} in Phase IV, and the received signals at $U_1$ and $E_1$ are, respectively,
\begin{align}
y_{{U_1}}^{FL_{Con}} &= \qh_{S,{U_1}}^T{{\bf{s}}_2} + n_{{U_1}},
\label{eqn:Uni User rec sig U1}
\\
y_{{E_1}}^{FL_{Con}} &= \qh_{S,{E_1}}^T{{\bf{s}}_2} +
n_{{E_1}}. \label{eqn:Con E1 FL sig}
\end{align}
Similarly, at the end of Phase VI, the received signals at $U_2$ and
$E_2$ are, respectively,
\begin{align}
y_{{U_2}}^{FL_{Con}} &= \qh_{S,{U_2}}^T{{\bf{s}}_1}+  n_{{U_2}},
\label{eqn:Uni User rec sig U2}
\\
y_{{E_2}}^{FL_{Con}} &= \qh_{S,{E_2}}^T{{\bf{s}}_1}+  n_{{E_2}}.
\label{eqn:Con E2 FL sig}
\end{align}

  {\rc The beamformer weights in the conventional scheme are exclusively designed at the GW for each user. Hence, when data is being transmitted for $U_1$, the satellite's main lobe is focused toward $U_1$. Since $E_2$ is outside the beam directed toward $U_1$ and the beamfomers are designed to maximize the signal strength toward $U_1$, $E_2$ receives the signal from side lobes. As a result, the signal received by $E_2$ is weak. Similar conditions hold for $E_1$ when transmitting to $U_2$.} To make the derivation tractable, {\rc we neglect these weak signals
received by $E_2$ and $E_1$ in Phases IV and VI, respectively}. As a result, the sum secrecy rate derived for the
conventional scheme shall be an upper-bound.
\subsubsection{Users' rates}
The RL rates for the conventional SATCOM are the same as the XOR
network coding scheme in~\eqref{eqn:U1 Uplink Rate}
and~\eqref{eqn:U2 Uplink Rate}. Using~\eqref{eqn:Uni User rec sig
U1} and~\eqref{eqn:Uni User rec sig U2}, the FL rates to $U_1$ and
$U_2$ after self-interference cancelation can be derived,
respectively, as
\begin{align}
I_{U_1}^{FL_{Con}} = t_2 \log _2  \left( 1 + \frac{\left| \qh_{S,{U_1}}^T  \qw_2 \right|^2}  { \sigma^2_{U_1}} \right),
\label{eqn:Con U1 FL rate}
\\
I_{U_2}^{FL_{Con}} = t_3 \log _2  \left( 1 + \frac{\left| \qh_{S,{U_2}}^T  \qw_1 \right|^2}  { \sigma^2_{U_2}} \right).
\label{eqn:Con U2 FL rate}
\end{align}
In order to make the conventional method comparable to the
bidirectional one, we assume that the total available transmission time
for both the network coding and the conventional schemes are the same. In
other words, the RL time for the users is $t_1$ and the FL for $U_1$
and $U_2$ are $t_2$ and $t_3= 1- t_1 - t_2$, respectively.
\subsubsection{Eavesdroppers' channel capacities}
The RL capacities for $E_1$ and $E_2$ in the conventional SATCOM are
the same as the ones derived for the XOR network coding scheme.
Using~\eqref{eqn:Con E1 FL sig} and~\eqref{eqn:Con E2 FL sig}, the FL capacity from the satellite toward $E_1$ and $E_2$ to overhear the signals sent in
Phases IV and VI, respectively, are
\begin{align}
I_{E_1}^{FL_{Con}}  &= t_2 \log _2  \left( 1 + \frac{\left| \qh_{S,{E_1}}^T  \qw_2 \right|^2}  { \sigma^2_{E_1}} \right),
\label{eqn:Con E1 FL rate}
\\
I_{E_2}^{FL_{Con}}  &= t_3 \log _2  \left( 1 + \frac{\left| \qh_{S,{E_2}}^T  \qw_1 \right|^2}  { \sigma^2_{E_2}} \right).
\label{eqn:Con E2 FL rate}
\end{align}
\subsubsection{Secrecy rate definition}
The RL secrecy rate for $U_1$ and $U_2$ are the same as the XOR network coding scheme 
in Section~\ref{sec:sec rate def bi}. {\rca In the conventional scheme, the messages that 
$E_1$ receives in the RL and FL are different and can be decoded independently. Hence,} the FL 
secrecy rate for $U_1$ can be defined using~\eqref{eqn:Con U1 FL rate},~\eqref{eqn:Con E1 FL rate} 
{\rca and the result from~\cite{Oggier:2008}} as
\begin{align}
{ SR_{{U_1}}^{F{L_{Con}}}} = I_{{U_1}}^{FL_{Con}} - I_{{E_1}}^{FL_{Con}}.
\label{eqn:Con secrecy rate U1 FL}
\end{align}
Utilizing~\eqref{eqn:XOR secrecy rate RL U1},~\eqref{eqn:Con secrecy rate U1 FL}{\rca , and 
Theorem~1 in~\cite{Awan:2012},} the end-to-end secrecy rate for $U_1$ is derived as
\begin{align}
{ SR_{{U_1}}^{Con} = \min \left\{ {SR_{{U_1}}^{RL},SR_{{U_2}}^{F{L_{Con}}}} \right\}.}
\label{eqn:Con secrecy rate end-to-end U1}
\end{align}
The end-to-end secrecy rate for $U_2$ can be defined in a similar
way. Like in~Section~\ref{sec:sec rate def bi}, the sum secrecy
rate is
\begin{align}
{ S{R^{Con}} = SR_{{U_1}}^{Con} + SR_{{U_2}}^{Con}.}
\label{eqn:Con sum secrecy rate}
\end{align}
\section{Problem Formulation and the Proposed Solution}\label{Sec:Problem Formulation}
In this section, we study the problem of maximizing the sum secrecy
rate by optimizing the precoding vectors at the GW to shape the
satellite beams along with the RL and FL time allocation, given the maximum
available power $P_S$ at the satellite. We consider both the XOR network coding and 
the conventional schemes. {\rca For the XOR network coding, we just solve the optimal beamformer design 
for the secrecy rate derived from the first case of the FL secrecy rate in~\eqref{eqn:FL SR}. The solutions for 
the optimal beamformer design for the other two cases of~\eqref{eqn:FL SR} are similar to the first case of~\eqref{eqn:FL SR}.}  
\subsection{Network coding for bidirectional SATCOM}
Using the sum secrecy rate defined in~\eqref{eqn:XOR sum secrecy
rate}, the optimization problem for the XOR network coding scheme is
defined as
\begin{align}
& \mathop {\max }\limits_{\qw,{t_1},{t_2}}  \,\,\,\, \min \left\{ { I_{{U_1}}^{RL} - I_{{E_1}}^{RL},{\rca \min \left\{ {I_{{U_1}}^{F{L_{XOR}}},I_{{U_2}}^{F{L_{XOR}}}} \right\}}} \right\} 
\nonumber\\
&+ \min \left\{ {I_{{U_2}}^{RL} - I_{{E_2}}^{RL},{\rca \min \left\{ {I_{{U_1}}^{F{L_{XOR}}},I_{{U_2}}^{F{L_{XOR}}}} \right\}}} \right\}   
\nonumber\\
& \,\,\,\, \mbox{s.t.}   \qquad  t_1+t_2=1,
\nonumber\\ 
&\,\,\,\,\,\,\,\,\,\, \qquad   \left\| \qw \right\|^2   \leq {P_S}.
\label{eqn:XOR Opt 1}
\end{align}

{\rca To transform~\eqref{eqn:XOR Opt 1} into a standard convex form, we apply the following procedures. First,} we assume that 
$t_1$ and $t_2$ are fixed and study the beamforming design. After designing the optimal beamformer, the optimal time
allocation is found by performing 1-D search of $t_1$ over the range $(0,1)$. {\rca Second, after considering a fixed transmission time for the RL and FL, the RL secrecy rate expressions in~\eqref{eqn:XOR Opt 1} are fixed and can be dropped 
without loss of generality. Hence,~\eqref{eqn:XOR Opt 1} boils down into
\begin{align}
& \mathop {\max }\limits_{\qw}  \,\,\,\, \min \left\{ {I_{{U_1}}^{F{L_{XOR}}},I_{{U_2}}^{F{L_{XOR}}}} \right\}  
\nonumber\\
& \, \mbox{s.t.}   \qquad  \left\| \qw \right\|^2   \leq {P_S}.
\label{eqn:XOR Opt 2}
\end{align}
Next, we introduce the auxiliary variable $u$ to remove the ``min'' 
operators. Then,~\eqref{eqn:XOR Opt 2} yields
\begin{align}
& \mathop {\max }\limits_{\qw,u>0}  \,\,\,\, u
\nonumber\\
& \,\,\, \mbox{s.t.}    \qquad   \left\| \qw \right\|^2   \leq {P_S},
\nonumber\\ 
& \,\,\,\,\,\,\,\,\,\,\, \qquad \sigma _{{U_1}}^2\left( {{2^{\frac{u}{{{t_2}}}}} - 1} \right) \le {\left| \qh_{S,{U_1}}^T  \qw_2 \right|^2},
\nonumber\\ 
& \,\,\,\,\,\,\,\,\,\,\, \qquad \sigma _{{U_2}}^2\left( {{2^{\frac{u}{{{t_2}}}}} - 1} \right) \le {\left| \qh_{S,{U_2}}^T  \qw_2 \right|^2}.
\label{eqn:XOR Opt 3}
\end{align}
The last two constraints in~\eqref{eqn:XOR Opt 3} are not convex. By introducing 
${\bf{W}} = {\bf{w}}{{\bf{w}}^H}$, we rewrite~\eqref{eqn:XOR Opt 3} as
\begin{align}
& \mathop {\max }\limits_{\qW\succeq 0,u>0}  \,\,\,\, u  
\nonumber\\
& \,\,\,\,\,\,\, \mbox{s.t.}    \qquad  \text{tr}\left( {\bf{W}} \right)  \leq {P_S},
\nonumber\\ 
&  \,\, \qquad \qquad \sigma _{{U_1}}^2\left( {{2^{\frac{u}{{{t_2}}}}} - 1} \right) \le \text{tr}\left( {{\bf{WA}}} \right),
\nonumber\\ 
&  \,\, \qquad \qquad \sigma _{{U_2}}^2\left( {{2^{\frac{u}{{{t_2}}}}} - 1} \right) \le \text{tr}\left( {{\bf{WB}}} \right),
\label{eqn:XOR Opt 4}
\end{align}
where ${\bf{A}} = {\bf{h}}_{S,{U_1}}^*{\bf{h}}_{S,{U_1}}^T$ and 
${\bf{B}} = {\bf{h}}_{S,{U_2}}^*{\bf{h}}_{S,{U_2}}^T$.} The rank constraint, $\text{rank}\left(
{\qW} \right) = 1$, in~\eqref{eqn:XOR Opt 4} is dropped. {\rca The optimal beamforming weight in~\eqref{eqn:XOR Opt 4} is 
designed for the FL transmission. However, since the RL secrecy rates, which can be bottlenecks for the total 
end-to-end secrecy rate, are not considered in~\eqref{eqn:XOR Opt 4}, extra unnecessary power at the satellite may be utilized. To fix this, one 
last constraint is added to~\eqref{eqn:XOR Opt 4} to get 
\begin{align}
& \mathop {\max }\limits_{{\bf{W}} \succeq 0,u>0}  \,\,\,\, u  
\nonumber\\
& \,\,\,\,\,\,\, \mbox{s.t.}    \qquad  \text{tr}\left( {\bf{W}} \right)  \leq {P_S},
\nonumber\\ 
&  \,\, \qquad \qquad \sigma _{{U_1}}^2\left( {{2^{\frac{u}{{{t_2}}}}} - 1} \right) \le \text{tr}\left( {{\bf{WA}}} \right),
\nonumber\\ 
&  \,\, \qquad \qquad \sigma _{{U_2}}^2\left( {{2^{\frac{u}{{{t_2}}}}} - 1} \right) \le \text{tr}\left( {{\bf{WB}}} \right),
\nonumber\\
&  \,\, \qquad \qquad u \le \max \left\{ {I_{{U_1}}^{RL} - I_{{E_1}}^{RL},I_{{U_2}}^{RL} - I_{{E_2}}^{RL}} \right\}.
\label{eqn:XOR Opt 5}
\end{align}}

Problem~\eqref{eqn:XOR Opt 5} is recognized as a SDP problem, thus convex and can be efficiently
solved. {\rc  According to {\rca Theorem~2.2} in~\cite{Wenbao}, when there are {\rca three} constraints 
on the matrix variable of a SDP problem such as~\eqref{eqn:XOR Opt 5}, existence
of a rank-1 optimal solution for $N_S>2$ is guaranteed}. {\rc Hence}, if the
solution to~\eqref{eqn:XOR Opt 5} happens not to be rank-one,
we can use {\rca Theorem~2.2} in~\cite{Wenbao} to recover the rank-one
optimal solution {\rc out of a non-rank-1 solution.}

{\rc 
According to~\cite{Boyd}, the complexity of  problem~\eqref{eqn:XOR Opt 5} is
\begin{align}
{\rm{O}}\left( {\left( {{\rca{3}} + N_S^2} \right){{\left( {\frac{{N_S^2\left( {N_S^2 + 1} \right)}}{2}} \right)}^3}} \right).
\label{eqn:com comp}
\end{align}
Solving~\eqref{eqn:XOR Opt 5} is accompanied along with a 1-D exhaustive search over the time variable $t$. We assume that the time variable is 
divided into $m$ bins between $0$ and $1$. The overall computational complexity for designing the beamformer 
for the XOR network coding scheme is {\rca $m$} times the complexity mentioned 
in~\eqref{eqn:com comp}. This is typically affordable since the optimization is performed at the GW on the ground.}
\subsection{Conventional SATCOM} \label{con:form}
According to the secrecy rate defined in~\eqref{eqn:Con sum secrecy
rate}, the optimization problem for the conventional scheme is
\begin{align}
&\mathop {\max }\limits_{{{\bf{w}}_1},{{\bf{w}}_2},t_1,t_2} \,\,\,
\min \left\{ {I_{{U_1}}^{RL} - I_{ {E_1}}^{RL},I_{{U_2}}^{FL_{Con}}
- I_{ {E_2}}^{FL_{Con}}} \right\}
\nonumber\\
&+ \min \left\{ {I_{{U_2}}^{RL} - I_{
{E_2}}^{RL},I_{{U_1}}^{FL_{Con}} - I_{{E_1}}^{FL_{Con}}} \right\}
\nonumber\\
& \,\,\,\,\,\,\,\,\,       \text{s.t.} \qquad   \left\| {\bf{w}}_1 \right\|^2 +
\left\| {\bf{w}}_2 \right\|^2 \le {P_S}. \label{eqn:Conv Opt 1}
\end{align}
Assume that the power  split between  the beamforming vectors
$\qw_1$ and $\qw_2$ is  $\beta {P_S}$ and $\left( {1 - \beta }
\right){P_S}$ where $\beta$ is a given parameter with $0 \leq \beta
\leq1 $. Using the parameter $\beta$, the beamforming vectors
$\qw_1$ and $\qw_2$ in the power constraint of~\eqref{eqn:Conv Opt
1} can be separated. Hence,~\eqref{eqn:Conv Opt 1} can be
rewritten as
\begin{align}
&\mathop {\max }\limits_{{\qw_1},{\qw_2},t_1,t_2} \,\,\,\,
\min \left\{ {I_{{U_1}}^{RL} - I_{ {E_1}}^{RL},I_{{U_2}}^{FL_{Con}}
- I_{ {E_2}}^{FL_{Con}}} \right\}
\nonumber\\
&+ \min \left\{ {I_{{U_2}}^{RL} - I_{
{E_2}}^{RL},I_{{U_1}}^{FL_{Con}} - I_{ {E_1}}^{FL_{Con}}} \right\}
\nonumber\\
& \,\,\,\,\,\,\,\,\,  \text{s.t.}   \qquad    \left\| \qw_1 \right\|^2   \le
\beta {P_S},
\nonumber\\
&  \,\,\,\, \qquad  \qquad  \left\| \qw_2 \right\|^2 \le \left( {1 - \beta }
\right){P_S}. \label{eqn:Conv Opt 2}
\end{align}
 The problem \eqref{eqn:Conv Opt 2} can be expanded as
\begin{align}
& \hspace{-0.1cm}\mathop {\max }\limits_{{\qw_1},{\qw_2},{t_1},{t_2}}
 \min \left\{ {  SR_{{U_1}}^{RL}, {t_2}\log \left( {\frac{{\sigma _{{E_2}}^2}}{{\sigma _{{U_2}}^2}}\frac{{\sigma _{{U_2}}^2 + |{\bf{h}}_{S,{U_2}}^T{{\bf{w}}_1}{|^2}}}{{\sigma _{{E_2}}^2 + |{\bf{h}}_{S,{E_2}}^T{{\bf{w}}_1}{|^2}}}} \right)  } \right\}
\nonumber\\
&+ \min \left\{ {  SR_{{U_2}}^{RL}, {t_3}\log \left( {\frac{{\sigma _{{E_1}}^2}}{{\sigma _{{U_1}}^2}}\frac{{\sigma _{{U_1}}^2 + |{\bf{h}}_{S,{U_1}}^T{{\bf{w}}_2}{|^2}}}{{\sigma _{{E_1}}^2 + |{\bf{h}}_{S,{E_1}}^T{{\bf{w}}_2}{|^2}}}} \right)  } \right\}
\nonumber\\
& \,\,\,\,\,\,\,\,\, \text{s.t.}   \qquad     \left\| \qw_1 \right\|^2 \le
\beta {P_S},
\nonumber\\
&   \,\,\,\,  \qquad  \qquad   \left\| \qw_2 \right\|^2 \le \left( {1 - \beta } \right){P_S}. 
\label{eqn:Conv Opt 3}
\end{align}
{\rca Before further simplifying~\eqref{eqn:Conv Opt 3}, we first mention 
the following theorem.}
{\rca \begin{thm}
If the achievable secrecy rate is strictly greater than zero, the
power constraints in~\eqref{eqn:Conv Opt 3} are active at the
optimal point ${\qw}_1^\star$ and ${\qw}_2^\star$, i.e., 
$\left\| \qw_1 \right\|^2 = \beta {P_S}$ and $\left\| \qw_2 \right\|^2 = \left( {1 - \beta }
\right){P_S}$. 
\label{thm:active}
\end{thm}
\begin{proof}
The proof is given in Appendix~\ref{app:active}.
\end{proof}}

Using Theorem~\ref{thm:active}, we can show that the constraints
in~\eqref{eqn:Conv Opt 3} are active which enables us to
write~\eqref{eqn:Conv Opt 3} as
\begin{align}
&\mathop {\max }\limits_{{\qw_1},{\qw_2},{t_1},{t_2}}
\,\,\,\, \min \left\{ {I_{{U_1}}^{RL} - I_{
{E_1}}^{RL},  {t_2}\log \left( {\frac{{\sigma _{{E_2}}^2}}{{\sigma _{{U_2}}^2}}\frac{{{\bf{w}}_1^H{{\bf{U}}_2}{{\bf{w}}_1}}}{{{\bf{w}}_1^H{{\bf{E}}_2}{{\bf{w}}_1}}}} \right)  } \right\}
\nonumber\\
&+ \min \left\{ {I_{{U_2}}^{RL} - I_{ {E_2}}^{RL}, {t_3}\log \left( {\frac{{\sigma _{{E_1}}^2}}{{\sigma _{{U_1}}^2}}\frac{{{\bf{w}}_2^H{{\bf{U}}_1}{{\bf{w}}_2}}}{{{\bf{w}}_2^H{{\bf{E}}_1}{{\bf{w}}_2}}}} \right)   } \right\}
\nonumber\\
& \,\,\,\,\,\,\,\, \text{s.t.}   \qquad     \left\| \qw_1 \right\|^2 =
\beta {P_S},
\nonumber\\
& \,\,\, \qquad  \qquad  \left\| \qw_2 \right\|^2 = \left( {1 - \beta }
\right){P_S}, \label{eqn:Conv Opt 4}
\end{align}
where $ {{\bf{U}}_1}\triangleq \frac{{\sigma _{{U_1}}^2}}{{\left( {1
- \beta } \right){P_S}}}{\bf{I}} +
\qh_{S,{U_1}}^*\qh_{S,{U_1}}^T, {{\bf{U}}_2} \triangleq
\frac{{\sigma _{{U_2}}^2}}{{ \beta {P_S}}}{\bf{I}} +
\qh_{S,{U_2}}^*\qh_{S,{U_2}}^T, {{\bf{E}}_1} \triangleq
\frac{{\sigma _{{E_1}}^2}}{{\left( {1 - \beta }
\right){P_S}}}{\bf{I}} +
\qh_{S,{E_1}}^*\qh_{S,{E_1}}^T,{{\bf{E}}_2} \triangleq
\frac{{\sigma _{{E_2}}^2}}{{\beta {P_S}}}{\bf{I}} +
\qh_{S,{E_2}}^* \qh_{S,{E_2}}^T.$

The benefit of \eqref{eqn:Conv Opt 4} is that given $\beta$, ${\bf w}_1$ and ${\bf w}_2$  can be optimized separately. To be specific, the optimal ${\bf w}_1$ and ${\bf w}_2$ corresponds to the eigenvectors associated with the maximum eigenvalues of matrices ${\bf{C}} = {\bf{L}}_1^{ - 1}{{\bf{U}}_1}{\bf{L}}_1^{ - H}$ and ${\bf{D}} = {\bf{L}}_2^{ - 1}{{\bf{U}}_2}{\bf{L}}_2^{ - H}$ where ${{\bf{E}}_1} = {{\bf{L}}_1}{\bf{L}}_1^H$ and ${{\bf{E}}_2} = {{\bf{L}}_2}{\bf{L}}_2^H$, respectively. As a result,~\eqref{eqn:Conv Opt 4} can be simplified into
\begin{align}
&\mathop  {\max }\limits_{\scriptstyle0 < t_1 < 1\hfill\atop \scriptstyle0 < t_2 < 1\hfill}  \,\, \min \left\{
{I_{U_1}^{RL} - I_{E_1}^{RL}, t_2\log \left( \frac{\sigma
_{{E_2}}^2}{\sigma _{U_2}^2}{\lambda _{\max}}\left( {\bf{C}} \right) \right)} \right\}
\nonumber\\
& + \min \left\{ {I_{U_2}^{RL} - I_{E_2}^{RL}, t_3 \log \left(
\frac{\sigma _{{E_1}}^2}{\sigma _{{U_1}}^2}{\lambda _{\max
}}\left( {\bf{D}}  \right) \right)} \right\}.
\label{eqn:Conv Opt 5}
\end{align}
Note that the constraints of~\eqref{eqn:Conv Opt 4} are dropped in~\eqref{eqn:Conv Opt 5} due to the 
homogeneity of the objective function. To solve~\eqref{eqn:Conv Opt 5}, we introduce auxiliary variables as
$u_1$ and $u_2$ to remove the ``$\min$'' operators as
\begin{align}
&\mathop {\max }\limits_{{t_1},{t_2},{u_1},{u_2}} \,\,\,\,  {u_1} +
{u_2}
\nonumber\\
& \,\,\,\,\,\,\, \text{s.t.}   \qquad    {u_1} \le {t_1}c,
\refstepcounter{equation} \subeqn \label{subeq:u1 1}
\nonumber\\
&    \qquad  \qquad   \,\,               {u_1} \le {t_2}\log \left(
{\frac{{\sigma _{{E_2}}^2}}{{\sigma _{{U_2}}^2}}{\lambda _{\max
}}\left( {\bf C} \right)} \right), \subeqn
\label{subeq:u1 2}
\nonumber\\
&    \qquad  \qquad    \,\,              {u_2} \le {t_1}d, \subeqn
\label{subeq:u2 1}
\nonumber\\
&    \qquad  \qquad      \,\,            {u_2} \le {t_3}\log \left(
{\frac{{\sigma _{{E_1}}^2}}{{\sigma _{{U_1}}^2}}{\lambda _{\max
}}\left( {\bf D}   \right)} \right), \subeqn
\label{subeq:u2 2}
\nonumber\\
&    \qquad  \qquad        \,\,          u_1,u_2 \geq 0,
\subeqn\label{subeq:u} \\& \qquad  \qquad \,\, 0<
{t_1}<1,0<{t_2}<1,\subeqn
\end{align}
where 
\begin{align}
c \triangleq \log \frac{1 +
\frac{P_{U_1}\|{\bf{h}}_{U_{1,S}}\|^2}{\sigma_S^2}}{{\left( {1 +
\frac{{{P_{{U_1}}}{{\left| {{h_{{U_1},{E_1}}}} \right|}^2}}}{{\sigma
_{{E_1}}^2}}} \right)}}, \,\,   d\triangleq\log \frac{1 +
\frac{P_{U_2}\|{\bf{h}}_{U_{2,S}}\|^2}{\sigma_S^2}}  {{\left( {1 +
\frac{{{P_{{U_2}}}{{\left| {{h_{{U_2},{E_2}}}} \right|}^2}}}{{\sigma
_{{E_2}}^2}}} \right)}}, \label{eqn:Conv Opt 6}
\end{align}
and $t_3 = 1 - t_1 - t_2$. Clearly, it is a linear programming
problem and can be optimally solved. After that, we use 1-D search
to find the optimal power allocation parameter $\beta^\star$.

\section{Simulation Results}\label{Sec:Simulation Results}
In this section, we present numerical results to evaluate the
secrecy rate of the XOR network coding based SATCOM protocol 
and compare it with the conventional scheme. We consider both i)
equal RL and FL time allocation (ETA), and ii)
optimized time allocation between the RL and the FL (OTA). We use
labels ``XOR-ETA'' and ``XOR-OTA'' to denote equal time allocation
and optimal time allocation policies, respectively.

In our simulations, $B$ denotes the carrier bandwidth, $41.67$
kHz, for both RL and FL transmissions. Since there is a main direct
link from the satellite to the users as well as some diffuse components, the
channel from the satellite to the users can be modeled as
Rician~\cite{Vucetic:1992}. The $K$-factor for the FL is determined
by the multipath average scattered power and random log-normal
variable using the values provided by~\cite{Vucetic:1992}. Due to
the ``scintillation'' effect~\cite{ITU:Uplink}, we have multipath in
the RL. Moreover, there exists a direct link like the FL case.
Therefore, the RL can be considered to follow Rician distribution with a higher
K-factor which is assumed to be $15$ dB. The rest of the link
parameters are summarized in Table~\ref{table:Link budget
parameters}~\cite{Fossa:1998}.  The satellite's FL transmission power in Table~\ref{table:Link budget
parameters} shows the carrier power used in the following transmissions: 1) the broadcast in Phase IV of the XOR scheme or, 2) the transmissions in 
Phases IV and V of the conventional reference scheme. If the satellite's FL transmission power is not a variable in a simulation scenario, its value provided by 
Table~\ref{table:Link budget parameters} is used.

The  ground channels between the users and the eavesdroppers are assumed to
follow a Rayleigh distribution with the pathloss calculated by
\begin{align}
{L} = 10\log \left[ {{{\left( {\frac{{4\pi }}{\lambda }}
\right)}^2}{d^\gamma }} \right], 
\label{eqn:Pathloss}
\end{align}
where $\gamma$ is the pathloss exponent which we assume to be $\gamma =3.7$.
The maximum Doppler shift is
calculated using the following relation as
\begin{align}
{f_D}_{max} = \frac{v}{\lambda } = \frac{{v{f_c}}}{c},
\label{eqn:Terrestrial Doppler Shift}
\end{align}
where $v$ is the user's speed, $f_c$ is the maximum frequency used
 and $c$ is the light speed.

Since the carrier bandwidth is $41.67$ kHz, we assume that the RL
operating bandwidth is $1616-1616.04167$ MHz for $U_1$,
$1616.04367-1616.08534$ MHz for $U_2$ and the FL operating
bandwidth is $1616-1616.04167$ MHz which is common between the
users. Each user is supposed to move in a random direction with a $10$
m/s speed. If not explicitly mentioned, each eavesdropper's distance
to the user is randomly changed between $2$ to $2.5$ km.

We first show  the average sum secrecy rate in Fig.~\ref{fig:ASR vs
Sat Ant} when   the number of feeds used on the satellite  varies
from $3$ to $10$.    As we can see, the XOR network coding scheme can achieve
over {\rca 54\%} higher average sum secrecy rate than the conventional one.
It can be observed that optimizing the RL and FL communication times
improves the average sum secrecy rate for both schemes considerably,
especially for the conventional scheme {\rca in higher number of feeds}.
\begin{table}[]
\caption{Link budget and parameters} 
\centering 
\begin{tabular}{|l| l|} 
\hline
Parameter & Value  \\ [0.5ex] 
\hline 
\hline
Satellite orbit type  & LEO  \\ 
\hline
Operating band (1$\sim$2 GHz) & L-band  \\
\hline
RL and FL frequency band, MHz   &  1616-1626.5       \\
\hline
Beams on the Earth                    &   48  \\
\hline
Number of antenna arrays           &  318  \\
\hline
Frequency reuse factor (FRF)         &  12\\
\hline
Number of carriers per beam               &   20   \\
\hline
Carrier bandwidth, $B_c$, kHz             &  41.67  \\
\hline
Guard bandwidth, kHz      &  2 \\
\hline
Satellite's antenna gain per beam, dBi    &  24.3 \\
\hline
Total power at the satellite, dBW         &  31.46 \\
\hline
Satellite noise temperature, K            & 290 \\
\hline
Terminal noise temperature, K             & 321 \\
\hline
Satellite's FL transmission power, dBW    &  7.65 \\
\hline
Mobile device radiation power, dBW        &  0 \\
\hline
Mobile device antenna gain, dBi           & $3.5$ \\
\hline
Return and forward link pathloss, dB      & 151 \\
\hline
Doppler shift due to satellite velocity, Hz   &  270 \\
\hline
Envelope mean of the direct wave, $m_s$ & 0.787     \\
\hline
The variance of the direct wave, $\sigma_s^2$  & 0.0671   \\
\hline
The power of the diffuse component &  0.0456    \\
\hline
\end{tabular}
\label{table:Link budget parameters}
\end{table}
\begin{figure}[]
  \centering
  \includegraphics[width=8.5cm]{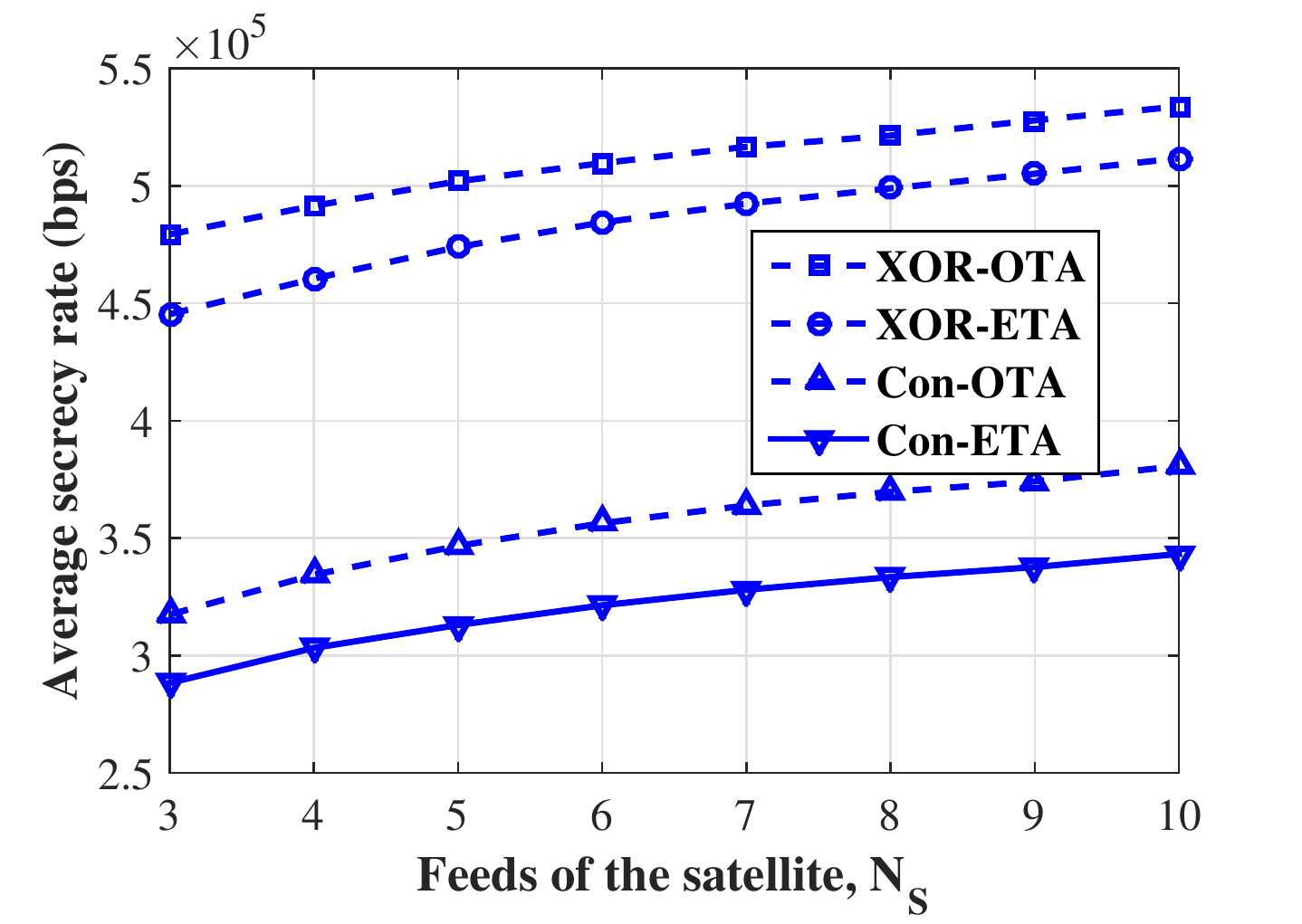}
  \caption{Average sum secrecy rate versus different number of feeds on the satellite for the XOR network coding and conventional
  schemes.}
  \label{fig:ASR vs Sat Ant}
\end{figure}
\begin{figure}[]
  \centering
  \includegraphics[width=8.5cm]{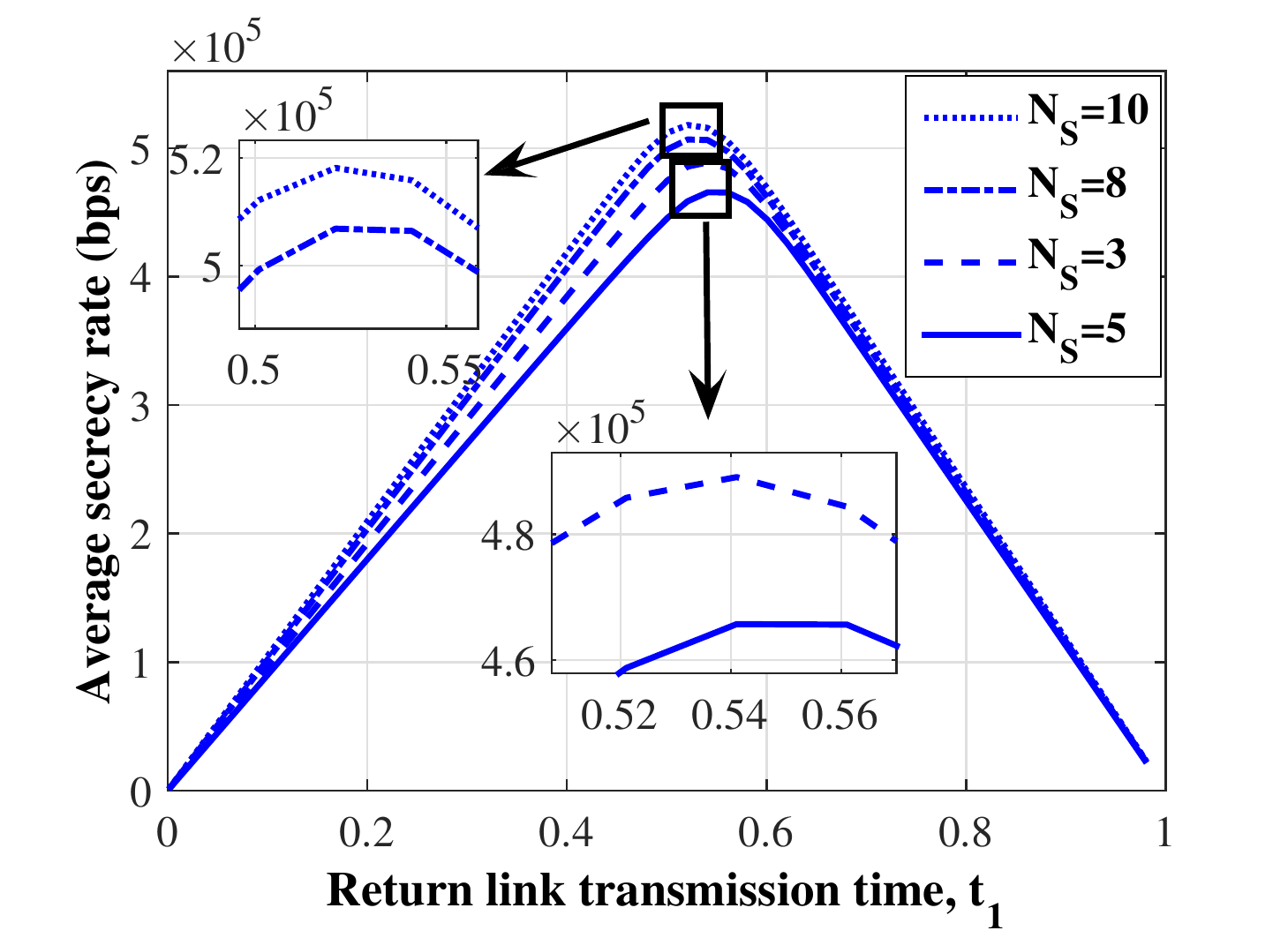}
  \caption{Average sum secrecy rate versus  the RL  time allocation $t_1$ in the XOR network coding scheme.}
  \label{fig:ASR vs uplink time-ant}
\end{figure}
\begin{figure}[]
  \centering
  \includegraphics[width=8.5cm]{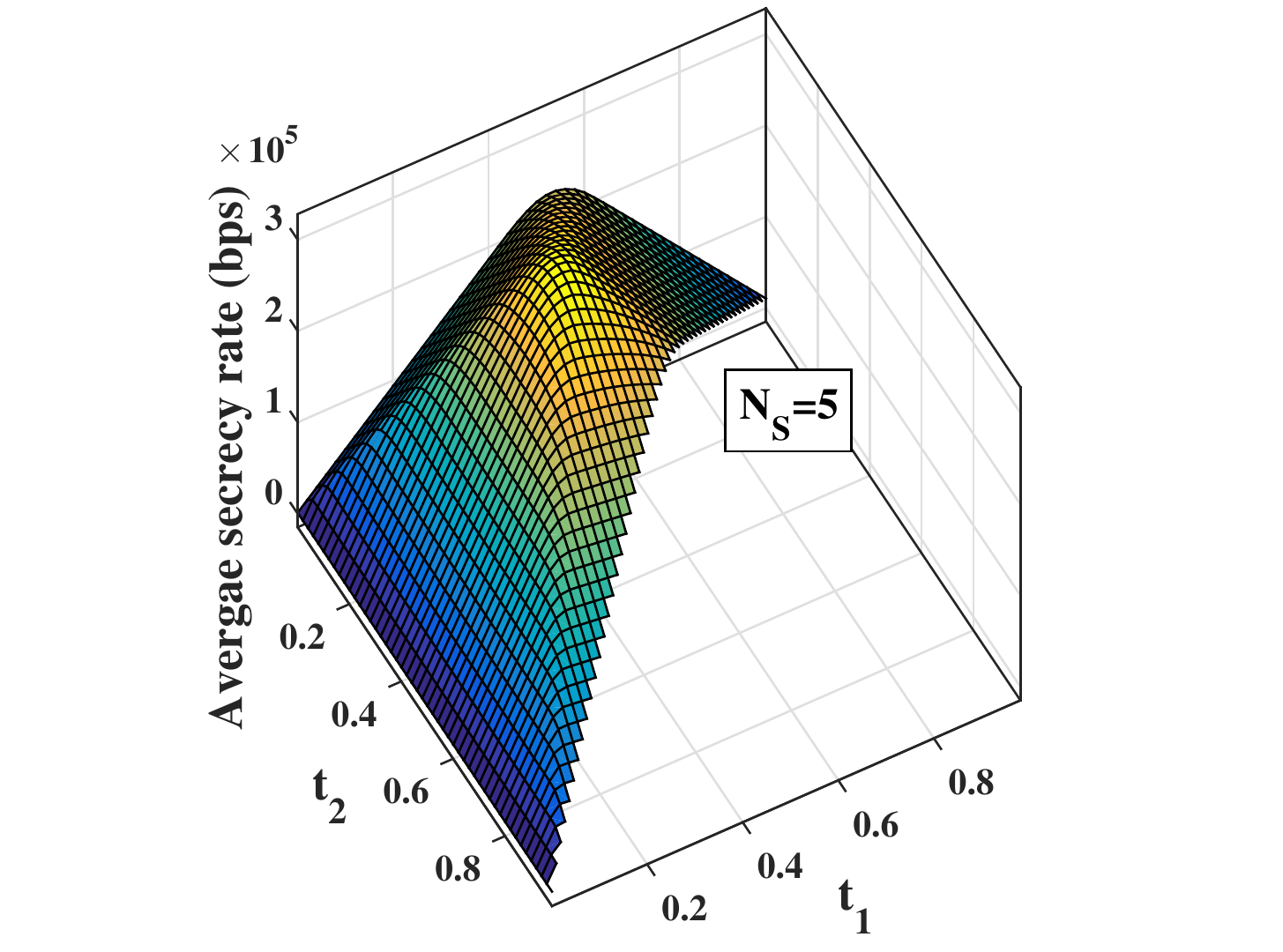}
  \caption{Average sum secrecy rate versus different RL, $t_1$, and FL, $t_2$ and $t_3=1-t_1-t_2$, time allocation in the conventional
  scheme.}
  \label{fig:ASR vs uplink time-con-ant}
\end{figure}
\begin{figure}[]
  \centering
  \includegraphics[width=8.5cm]{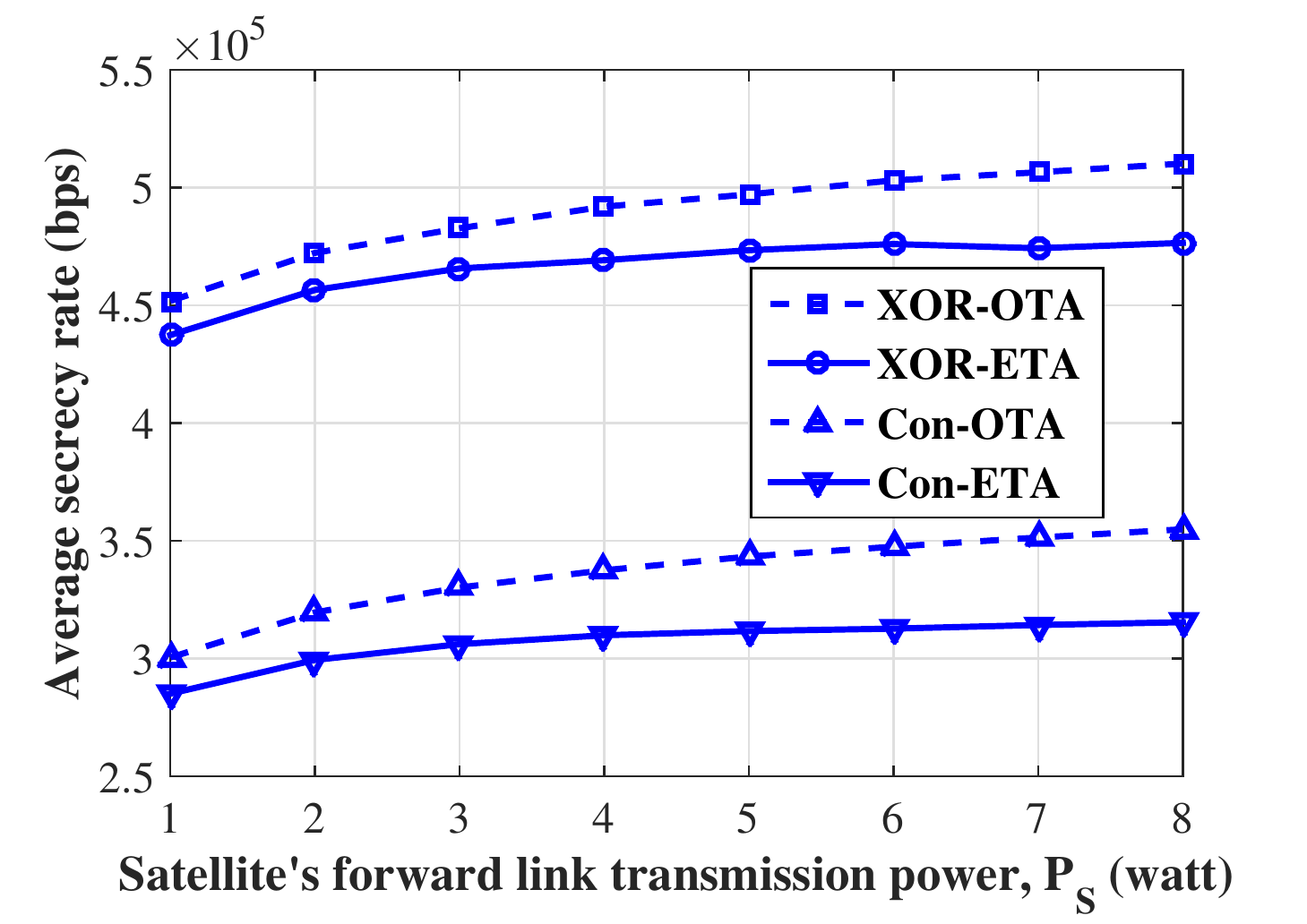}
  \caption{Average sum secrecy rate versus the satellite's forward link transmission power.}
  \label{fig:ASR vs Sat Pow}
\end{figure}
\begin{figure}[]
  \centering
  \includegraphics[width=8.5cm]{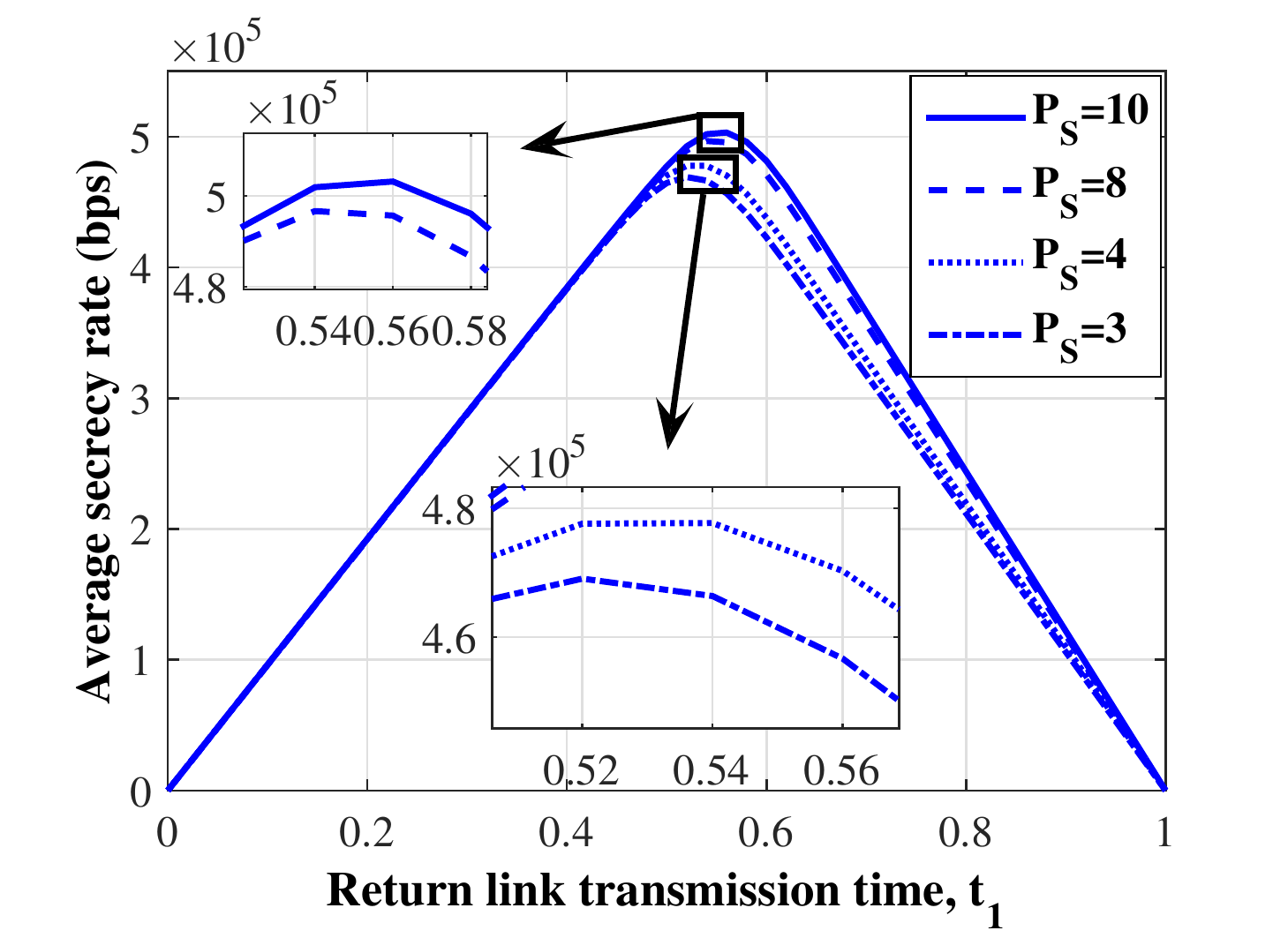}
  \caption{Average sum secrecy rate versus RL time allocation for different satellite's forward link transmission powers.}
  \label{fig:ASR vs uplink time-pow}
\end{figure}
\begin{figure}[]
  \centering
  \includegraphics[width=8.5cm]{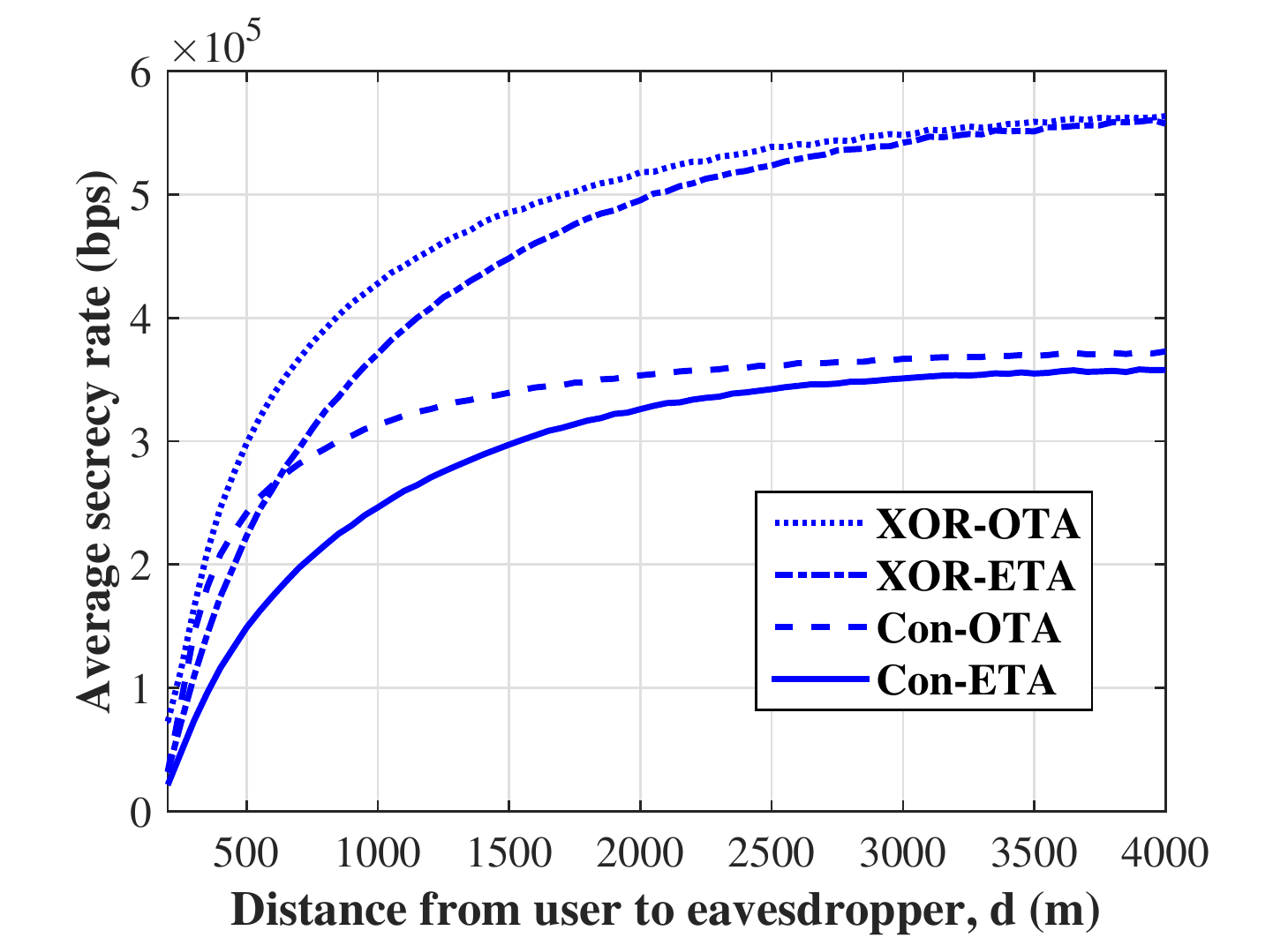}
  \caption{Average sum secrecy rate versus the distance between the user and the eavesdropper for XOR 
	network coding and conventional schemes while equal and optimal time allocation are employed.}
  \label{fig:ASR vs Sat dis}
\end{figure}
\begin{figure}[]
  \centering
  \includegraphics[width=8.5cm]{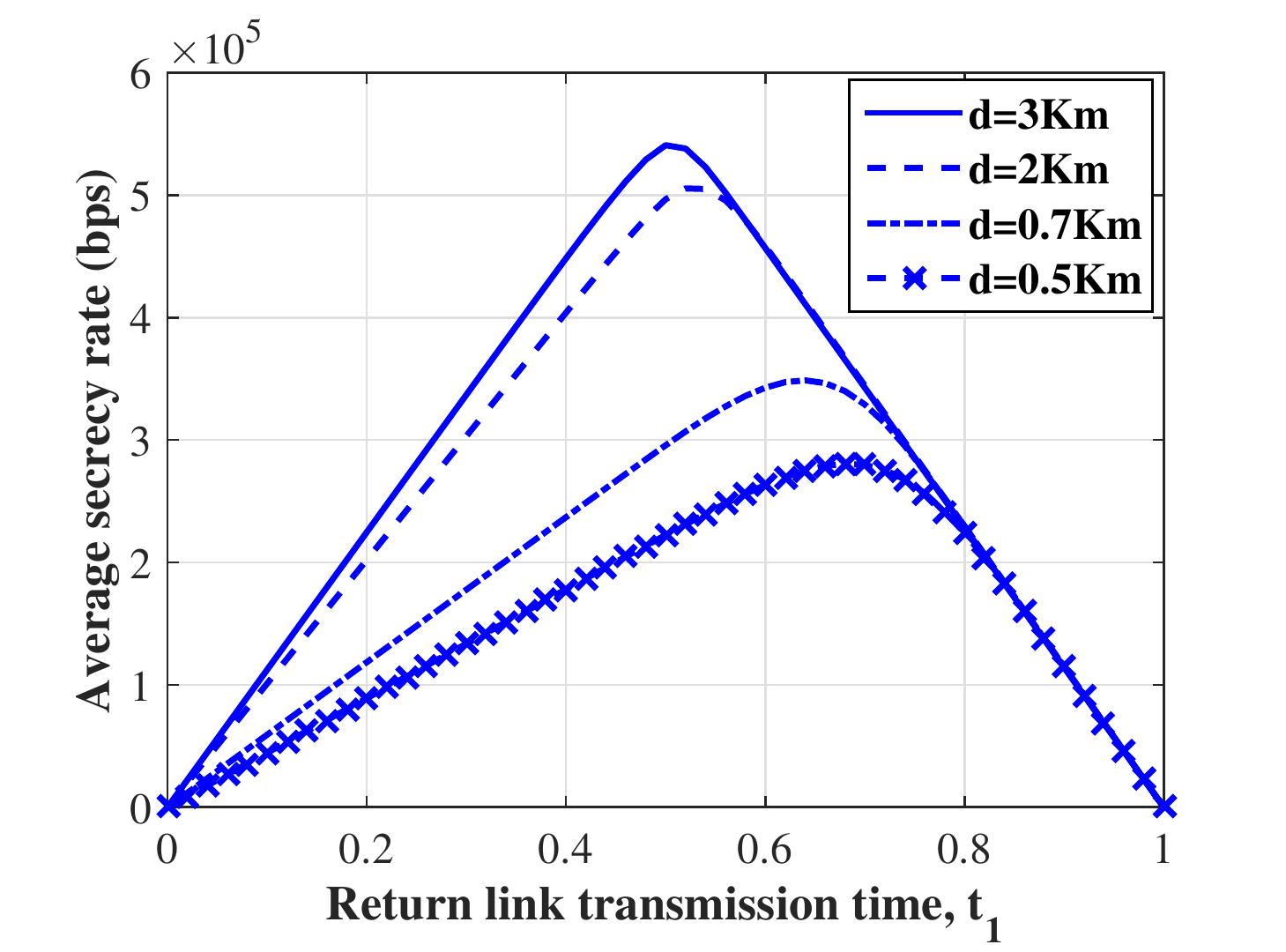}
  \caption{Average sum secrecy rate versus different RL and FL time allocation in XOR network coding scheme
  for different distances between the user and eavesdropper.}
  \label{fig:ASR vs uplink time-d}
\end{figure}

The effect of time allocation    is further illustrated in
Figs.~\ref{fig:ASR vs uplink time-ant} and~\ref{fig:ASR vs uplink
time-con-ant} for the XOR network coding and the
conventional schemes, respectively. It is observed in
Fig.~\ref{fig:ASR vs uplink time-ant} that for different number of
feeds, the average sum secrecy rate first increases, and then then
decreases with the RL time allocation $t_1$. Here, more time is allocated to the 
RL transmission which means that the FL transmission rate is a bottleneck for the end-to-end rate. The time split between the RL and FL 
depends on the number of feeds at the satellite. As the number of feeds increases, the time 
devoted to the FL transmission increases. This shows that the FL acts as a bottleneck for 
the end-to-end communications. The change in the RL
and FL time allocation makes the
channel secrecy rates closer to each other so that the overall
average secrecy rate increases. The optimal time allocation for one
RL slot and two FL slots in the conventional scheme can be seen in
Fig.~\ref{fig:ASR vs uplink time-con-ant}.

The effect of the satellite's FL transmission power on the average secrecy rate is
investigated in Figs.~\ref{fig:ASR vs Sat Pow} and~\ref{fig:ASR vs
uplink time-pow}. In Fig.~\ref{fig:ASR vs Sat Pow}, we see that the
average secrecy rate for the equal time allocation approach in both
schemes starts to saturate as the available power for the FL transmission increases. This can be
explained by the fact that as the available power increases,  the RL
becomes a bottleneck for the end-to-end secrecy rate and hinders the overall
improvement. On the other hand, while performing optimal time
allocation over RL and FL, the average secrecy rate keeps growing 
for both the conventional and the XOR network coding
schemes. It is seen in
Fig.~\ref{fig:ASR vs uplink time-pow} that by increasing the power
at the satellite, more time is allocated to the RL transmission in
order to  balance the RL and FL secrecy rates and sustaining the 
secrecy rate growth. However, after increasing the satellite's power beyond 
a specific point, the effect of the optimal time allocation fades out, and the average 
secrecy rate in the optimal time allocation scheme also saturates 
due to RL being a bottleneck. This fact can be observed in 
Fig.~\ref{fig:ASR vs uplink time-pow}. As the power of the FL transmission increases, less 
time is exchanged between the RL and FL transmission and the average secrecy rate saturates.

The effect of the distance between each user and the corresponding
eavesdropper is investigated in Figs.~\ref{fig:ASR vs Sat dis}
and~\ref{fig:ASR vs uplink time-d}. As  is seen in Fig.~\ref{fig:ASR
vs Sat dis}, the average secrecy rate for equal time allocation in both 
schemes saturates as the distance between the user and
eavesdropper increases. This is because increasing the distance to
the eavesdropper improves the secrecy rate in the RL, leaving
the FL as a performance bottleneck. When the time allocation
is optimized, the average secrecy rate shows notable gain in both schemes. However, after a specific distance, the secrecy rate
for the optimal power allocation also saturates. Increasing the distance to
the eavesdropper increases the secrecy rate for the RL, but this increment is going to be quite small at some point and consequently vanishes. {\rca Consequently, as the distance increases, less time exchange is required between the RL and FL transmission. This fact can be seen in Fig.~\ref{fig:ASR vs uplink time-d}. Due to this limit in the RL secrecy 
rate, the secrecy rate can be improved using optimal time allocation up to a limited distance. Furthermore, as it is observed in Fig.~\ref{fig:ASR vs Sat dis}, the average sum secrecy rate of the XOR network coding saturates in a much longer distance compared to the conventional scheme.} Interestingly, when the user and the eavesdropper are close, the conventional scheme using
the optimal time allocation outperforms the XOR network coding
scheme using equal time allocation. This originates from the fact that {\rca there are more degrees of freedom in terms of optimal time allocation in the conventional scheme compared to the XOR network coding scheme.} Hence, when it comes to picking up a secure protocol, distance plays an important role.

The results in Fig.~\ref{fig:ASR vs uplink time-d} illustrate that as the distance
between the user and the eavesdropper decreases, more time is
allocated to the RL transmission of the XOR network coding scheme in
order to balance the secrecy rates in RL and FL. It is observed that as the distance to the
eavesdropper increases, less change is required in the RL and FL times. This is due to the fact that as 
the distance increases, the improvement rate in the secrecy rate of the RL is reduced and less regulation is required in the transmission times.
\section{Conclusion}\label{sec:con}
Network coding principle has been known to increase the throughput
of bidirectional SATCOM. In this paper, we studied the
use of XOR network coding to improve the sum secrecy rate of
bidirectional SATCOM. The beamforming vector as well as the optimal
time allocation between the RL and the FL were optimized to improve
the secrecy rate.  We compared the sum secrecy rate of the XOR network coding with the conventional scheme without using
network coding regarding realistic system parameters. Our results
demonstrated that the network coding based scheme
outperforms the conventional scheme substantially, especially when
the legitimate users and the eavesdroppers are not close. 
\appendices
{\rca \section{Proof of Theorem~\ref{thm:active}} \label{app:active}
\begin{proof}
In the objective function of problem~\eqref{eqn:Conv Opt 3}, only the second argument of the ``min'' operators, FL secrecy rates, 
include the beamforming vector. Hence, we focus on these terms in our optimization. Using contradiction, we shall show that $\left\| {\bf{w}}_1^\star \right\|^2 = \beta {P_S}$ and 
$\left\| {\bf{w}}_2^\star \right\|^2 = \left( {1 - \beta }
\right){P_S}$ must hold for the optimal solutions ${\bf{w}}_1^\star$ and ${\bf{w}}_2^\star$. Assume that ${\bf{w}}_1^{\star}$ and ${\bf{w}}_2^{\star}$ are the optimal
solutions to~\eqref{eqn:Conv Opt 3} and satisfy $\left\| \qw_1 \right\|^2 < \beta {P_S}$ and 
$\left\| \qw_2 \right\|^2 < \left( {1 - \beta }
\right){P_S}$, then there exist constants $\alpha_1>1$ and $\alpha_2>1$ that satisfy ${\left\| {{{\widehat {\bf{w}}}_1^{\star}}} \right\|^2} = \beta {P_S}$ 
and ${\left\| {{{\widehat {\bf{w}}}_2^{\star}}} \right\|^2} = \left( {1 - \beta }\right){P_S}$ where 
${{{\widehat {\bf{w}}}_1^{\star}}} = \alpha_1  {\qw_1^ \star }$ and ${{{\widehat {\bf{w}}}_2^{\star}}} = \alpha_2  {\qw_2^ \star }$. 
Replacing ${\qw_1^ \star }$ by ${\widehat {\bf{w}}}_1^{\star}$ and ${\qw_2^ \star }$ by ${\widehat {\bf{w}}}_2^{\star}$ 
in the FL secrecy rates of the objective in~\eqref{eqn:Conv Opt 3}, we get
\begin{align}
&  {f_1}\left( {{\alpha _1}} \right) = {t_2}\log \left( {\frac{{\sigma _{{E_2}}^2}}{{\sigma _{{U_2}}^2}}\frac{{\sigma _{{U_2}}^2 + \alpha_1^2|{\bf{h}}_{S,{U_2}}^T{\bf{w}}_1^ \star {|^2}}}{{\sigma _{{E_2}}^2 + \alpha_1^2 |{\bf{h}}_{S,{E_2}}^T{\bf{w}}_1^ \star {|^2}}}} \right),
\nonumber\\
& {f_2}\left( {{\alpha _2}} \right) = {t_3}\log \left( {\frac{{\sigma _{{E_1}}^2}}{{\sigma _{{U_1}}^2}}\frac{{\sigma _{{U_1}}^2 + \alpha_2^2|{\bf{h}}_{S,{U_1}}^T{\bf{w}}_2^ \star {|^2}}}{{\sigma _{{E_1}}^2 + \alpha_2^2|{\bf{h}}_{S,{E_1}}^T{\bf{w}}_2^ \star {|^2}}}} \right).
\label{eqn:XOR Opt Simp act}
\end{align}
Also, we assume that in the RL and FL of each user the secrecy rate is nonzero which translates into
\begin{align}
\sigma _{{E_2}}^2\left( {\sigma _{{U_2}}^2 + |{\bf{h}}_{S,{U_2}}^T{{\bf{w}}_1}{|^2}} \right) > \sigma _{{U_2}}^2\left( {\sigma _{{E_2}}^2 + |{\bf{h}}_{S,{E_2}}^T{{\bf{w}}_1}{|^2}} \right), 
\exists \qw_1, 
\label{eqn:Pos1}
\\
\sigma _{{E_1}}^2\left( {\sigma _{{U_1}}^2 + |{\bf{h}}_{S,{U_1}}^T{{\bf{w}}_2}{|^2}} \right) > \sigma _{{U_1}}^2\left( {\sigma _{{E_1}}^2 + |{\bf{h}}_{S,{E_1}}^T{{\bf{w}}_2}{|^2}} \right), 
\exists \qw_2. 
\label{eqn:Pos2}
\end{align}
According to the conditions in~\eqref{eqn:Pos1} and~\eqref{eqn:Pos2}, we can see that ${f_1}(\alpha)$ and ${f_2}(\alpha)$ 
are monotonically increasing functions in the parameters $\alpha_1$ and $\alpha_2$. This contradicts 
that ${\bf{w}}_1^ \star$ and ${\bf{w}}_2^ \star$ are the optimal solutions. Since adjusting the RL 
and FLs transmission time balances the RL and FL secrecy rates, the RL bottleneck does not limit the FL secrecy rate increment. 
Hence, the power constraint should be active. This completes the proof.
\end{proof}}

\begin{IEEEbiography}
    [{\includegraphics[width=1in,height=1.25in,clip,keepaspectratio]{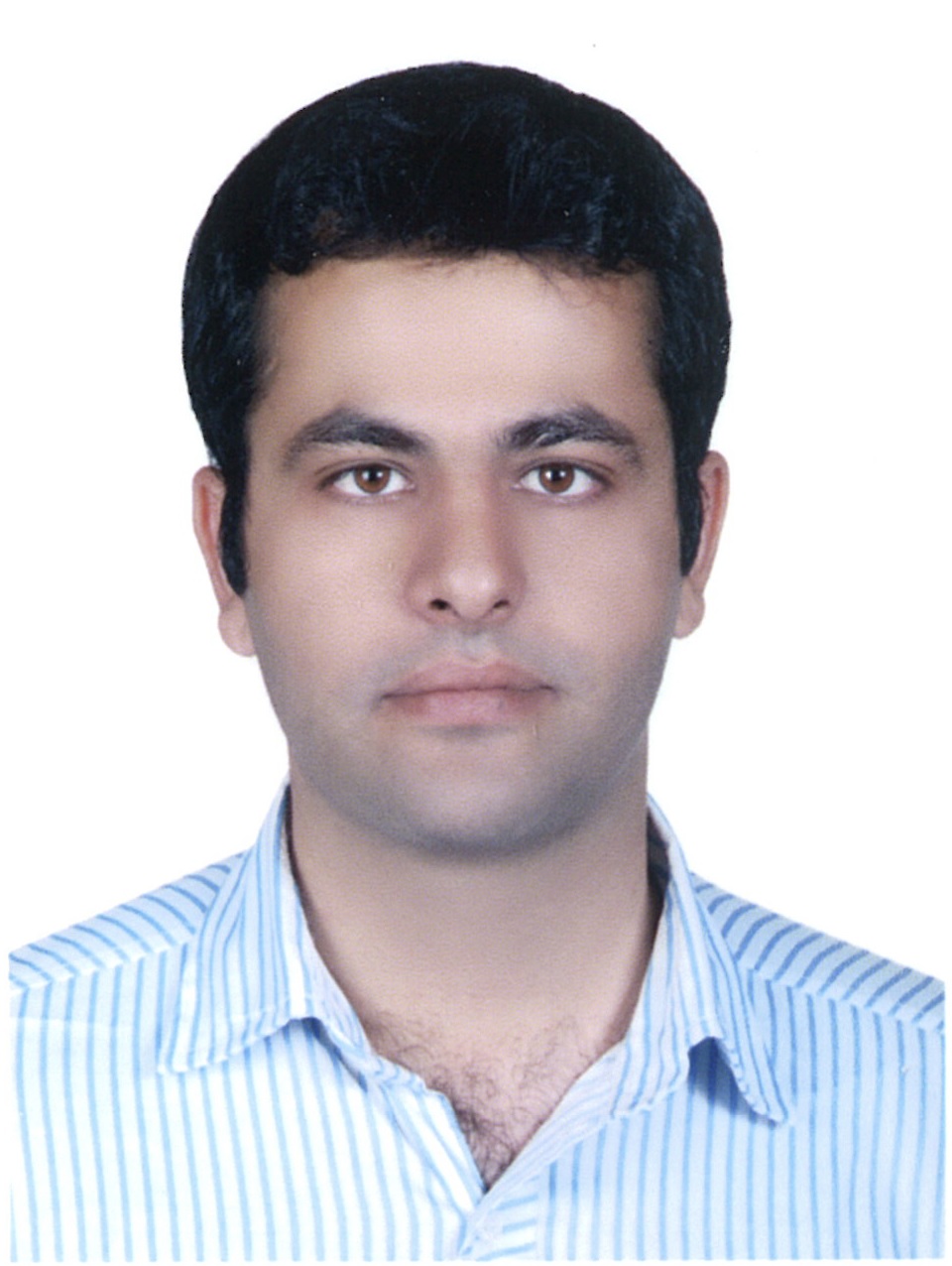}}]{Ashkan Kalantari}
Ashkan Kalantari (AK) was born in Iran. He received his BSc and MSc degrees from K. N. Toosi 
University of Technology, Tehran, Iran in 2009 and 2012, respectively. He is currently working toward the 
Ph.D. degree with the research group SIGCOM in the Interdisciplinary Centre for Security, Reliability and 
Trust (SnT), University of Luxembourg. His research interest is physical layer security in wireless and satellite communications.
\end{IEEEbiography}
\begin{IEEEbiography}
    [{\includegraphics[width=1in,height=1.25in,clip,keepaspectratio]{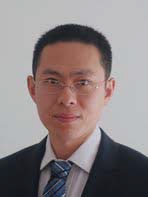}}]{Gan Zheng} 
		(S'05-M'09-SM'12) 
is currently a Lecturer in School of Computer Science
and Electronic Engineering, University of Essex, UK. 
He received the B. E. and the M. E. from Tianjin University, Tianjin, China, in 2002 and 2004,
respectively, both in Electronic and Information Engineering, and the PhD degree in Electrical 
and Electronic Engineering from The University of Hong Kong, Hong Kong, in 2008. Before he joined University of Essex, he worked as a Research 
Associate at University College London, UK, and University of Luxembourg,
Luxembourg. His research interests include cooperative communications, cognitive radio, physical-layer security, full-duplex radio and 
energy harvesting. He is the first recipient for the 2013 IEEE Signal Processing Letters Best Paper Award.
\end{IEEEbiography}
\begin{IEEEbiography}
    [{\includegraphics[width=1in,height=1.25in,clip,keepaspectratio]{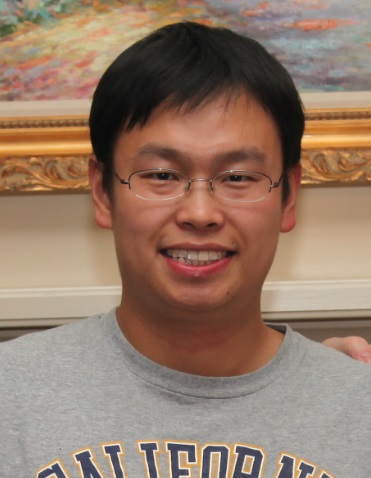}}]{Zhen GAO} 
received his BS, MS and PhD degree in Electrical and Information Engineering from Tianjin University, China, in 2005, 2007 and 2011, respectively. From 2008.10 to 2010.11, he was a visiting scholar in GeorgiaTech, working on the design and implementation for OFDM based cooperative communication. From 2011.7 to 2014.11, he was an assistant researcher in the Wireless and Mobile Communication Research Center in Tsinghua University, China. He is currently an Associate Professor in Tianjin University. His focus is on mobile satellite communications, fault-tolerant signal processing and wireless communication system.
\end{IEEEbiography}	
\begin{IEEEbiography}
    [{\includegraphics[width=1in,height=1.25in,clip,keepaspectratio]{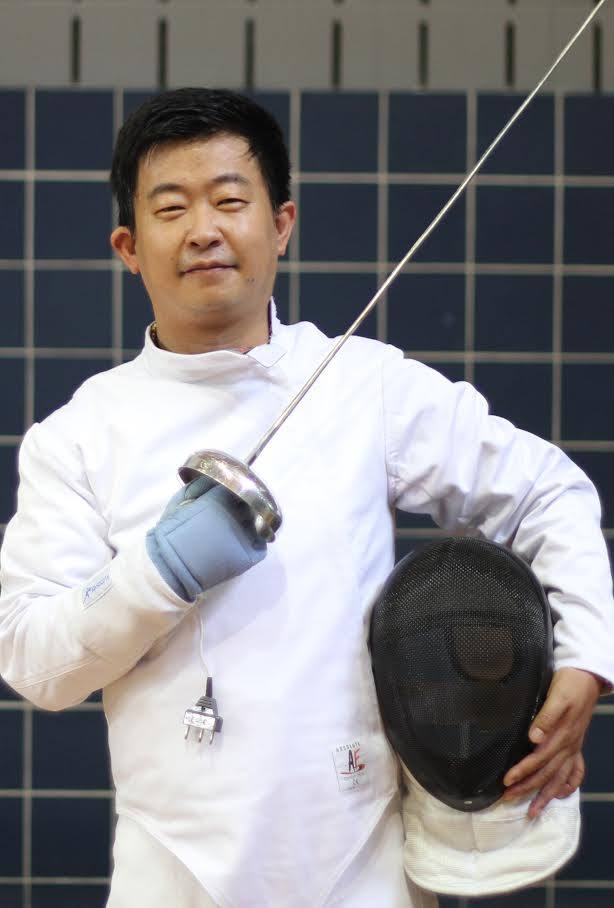}}]{Zhu Han}
 (S’01–M’04-SM’09-F’14) received the B.S. degree in electronic engineering from Tsinghua University, in 1997, and the M.S. and Ph.D. degrees in electrical engineering from the University of Maryland, College Park, in 1999 and 2003, respectively. 

From 2000 to 2002, he was an R\&D Engineer of JDSU, Germantown, Maryland. From 2003 to 2006, he was a Research Associate at the University of Maryland. From 2006 to 2008, he was an assistant professor in Boise State University, Idaho. Currently, he is an Associate Professor in Electrical and Computer Engineering Department at the University of Houston, Texas. His research interests include wireless resource allocation and management, wireless communications and networking, game theory, wireless multimedia, security, and smart grid communication. Dr. Han is an Associate Editor of IEEE Transactions on Wireless Communications since 2010. Dr. Han is the winner of IEEE Fred W. Ellersick Prize 2011. Dr. Han is an NSF CAREER award recipient 2010. Dr. Han is IEEE Distinguished lecturer since 2015.  
\end{IEEEbiography}
\begin{IEEEbiography}
    [{\includegraphics[width=1in,height=1.25in,clip,keepaspectratio]{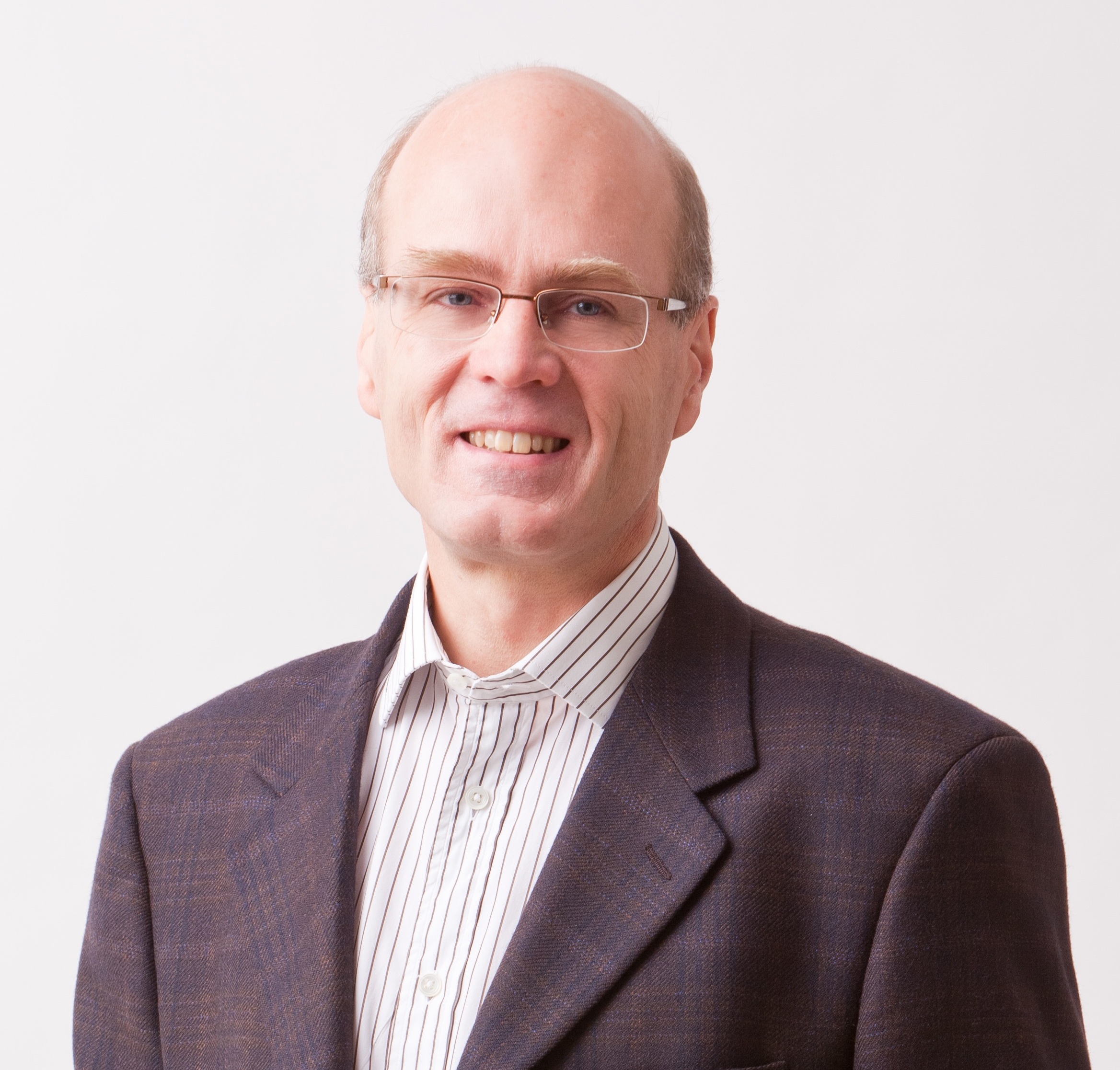}}]{Bj\"{o}rn Ottersten}
was born in Stockholm, Sweden, 1961. He received the M.S. degree in electrical engineering and applied physics from Link\"{o}ping University, Link\"{o}ping, Sweden, in 1986. In 1989 he received the Ph.D. degree in electrical engineering from Stanford University, Stanford, CA. Dr. Ottersten has held research positions at the Department of Electrical Engineering, Link\"{o}ping University, the Information Systems Laboratory, Stanford University, the Katholieke Universiteit Leuven, Leuven, and the University of Luxembourg. During 96/97 Dr. Ottersten was Director of Research at ArrayComm Inc, a start-up in San Jose, California based on Ottersten’s patented technology. He has co-authored journal papers that received the IEEE Signal Processing Society Best Paper Award in 1993, 2001, 2006, and 2013 and 3 IEEE conference papers receiving Best Paper Awards. In 1991 he was appointed Professor of Signal Processing at the Royal Institute of Technology (KTH), Stockholm. From 1992 to 2004 he was head of the department for Signals, Sensors, and Systems at KTH and from 2004 to 2008 he was dean of the School of Electrical Engineering at KTH. Currently, Dr. Ottersten is Director for the Interdisciplinary Centre for Security, Reliability and Trust at the University of Luxembourg. Dr. Ottersten is a board member of the Swedish Research Council and as Digital Champion of Luxembourg, he acts as an adviser to the European Commission. Dr. Ottersten has served as Associate Editor for the IEEE Transactions on Signal Processing and on the editorial board of IEEE Signal Processing Magazine. He is currently editor in chief of EURASIP Signal Processing Journal and a member of the editorial boards of EURASIP Journal of Applied Signal Processing and Foundations and Trends in Signal Processing. Dr. Ottersten is a Fellow of the IEEE and EURASIP and a member of the IEEE Signal Processing Society Board of Governors. In 2011 he received the IEEE Signal Processing Society Technical Achievement Award. He is a first recipient of the European Research Council advanced research grant. His research interests include security and trust, reliable wireless communications, and statistical signal processing.
\end{IEEEbiography}
\end{document}